%% file: main.tex
\documentclass[acmtog]{acmart}

\usepackage{amsmath} 
\usepackage{amsthm}
\usepackage{bbm}
\usepackage{scalerel}

\usepackage{xcolor}
\usepackage{tikz}
\usepackage{pgfplots}
\usetikzlibrary{automata, positioning, arrows}

\newcommand{\red}{\textcolor{red}} 
 
\newcommand{\dgreen}{\textcolor{black!40!green}}

\newcommand{\cB}{\mathcal{B}} 
\newcommand{\cD}{\mathcal{D}} 
\newcommand{\cX}{\mathcal{X}} 
\newcommand{\cS}{\mathcal{S}} 
 
\newcommand{\cY}{\mathcal{Y}} 
\newcommand{\cE}{\mathcal{E}} 
\newcommand{\cH}{\mathcal{H}} 
\newcommand{\bP}{\mathbb{P}}
\newcommand{\oell}{\overline{\ell}}

\newcommand{\ty}{\mathtt{y}}

\newcommand{\slca}{S$\ell$-CA}

\DeclareMathOperator{\Pre}{Pre}

\newtheorem{definition}{Definition}
\newtheorem{proposition}{Proposition}

\newtheorem{theorem}{Theorem}
\newtheorem{assumption}{Assumption}
\newtheorem{corollary}{Corollary}
\theoremstyle{remark}
\newtheorem{remark}{Remark}
\newtheorem{example}{Example}
\newtheorem*{proof}{Proof}
\usepackage{todonotes}

\begin{document}

% front page
\title{Data-driven Abstractions for Verification of Deterministic Systems}
% \author{Rudi Coppola, Andrea Peruffo, Manuel Mazo Jr.
% \thanks{This work was supported by the European Research Council through the SENTIENT project (ERC-2017-STG \#755953).}% <-this % stops a space
% \thanks{R. Coppola, A. Peruffo and M. Mazo are with Faculty of Mechanical, Maritime and Materials Engineering, TU Delft, Delft, The Netherlands.
% {\tt\small r.coppola@tudelft.nl, a.peruffo@tudelft.nl}}
% }

\author{Rudi Coppola}
\authornote{Both authors contributed equally to this research.}
\email{r.coppola@tudelft.nl}
\orcid{0000-0002-1876-6827}

\author{Andrea Peruffo}
\authornotemark[1]
\email{a.peruffo@tudelft.nl}
\orcid{0000-0002-7767-2935}

\author{Manuel Mazo Jr.}
\email{m.mazo@tudelft.nl}
\orcid{0000-0002-5638-5283}

\affiliation{%
  \institution{TU Delft}
  \city{Delft}
  \country{Netherlands}
}

% \author{Rudi Coppola}
% \email{}
% \affiliation{%
%   \institution{Institution}
%   \country{Country}

\begin{abstract}
    A common technique to verify complex logic specifications for dynamical systems is the construction of symbolic abstractions: simpler, finite-state models whose behaviour mimics the one of the systems of interest. 
    Typically, abstractions are constructed exploiting an accurate knowledge of the underlying model: in real-life applications, this may be a costly assumption. 
    By sampling random $\ell$-step trajectories of an unknown system, we build an abstraction based on the notion of $\ell$-completeness. 
    We newly define the notion of probabilistic behavioural inclusion, and provide probably approximately correct (PAC) guarantees that this abstraction includes all behaviours of the concrete system, 
    for finite and infinite time horizon, 
    leveraging the scenario theory for non convex problems. 
    Our method is then tested on several numerical benchmarks.
\end{abstract}

\maketitle

%%%%%%%%%%%%%%%%%%%%%%%%%%%%%%%%%%%%%%%%%%%%%%%%%%%%%%

\section{Introduction}
\label{sec:intro}

Data-driven modeling and analysis is living a new renaissance with the advances in machine-learning and artificial intelligence enabled by unprecedented computing power. The field of system's verification, aiming at providing formal performance and safety guarantees is not alien to this trend. Recent work has been focused on the use of collected data from a system to derive directly (with no model involved) barrier functions certifying invariance~\cite{akella2022barrier, salamati2021data}, or finite abstractions to verify and synthesize controllers~\cite{devonport2021symbolic, majumdar2020abstraction, sadraddini2018formal}. A popular approach is to employ scenario-based optimization techniques to derive probably approximately correct (PAC) guarantees on the performance metric of interest.

% However, 
The use of scenario-based optimization 
%imposes certain limitations on how the data must be generated. In particular, each data point needs to be sampled independently  
requires independent samples, generated from the probability distribution that drives the system's uncertainty.
A (simpler) special case considered in many works \cite{makdesi2022data, devonport2021symbolic, lavaei2022data, wang2020data, wang2021scenario} is that of deterministic systems (for which a model is not available), in which the only uncertainty is their initialization, i.e. the initial state is drawn from some probability distribution and determines the whole trajectory of the system.  
In such a system the independence of samples can be derived by sampling the initial state in an independent fashion, 
% following the assumed distribution on initial states, 
e.g. typically a uniform distribution is selected for compact initial sets. 
However, when employing data to construct finite abstractions, or verifying barrier functions, the necessary data are individual transitions of the system: 
thus, to aquire independent samples of transitions implies acquiring samples from independent trajectories.
This means that to collect data to \emph{learn} e.g. a barrier function for a set larger than the initial one, one may have to discard large portions of trajectories to retain independence of samples.  
%
% Furthermore, in the constructions of finite abstractions, or certificates (like barrier functions), based on 
By using 
one-step transitions, the scenario-based approach provides guarantees for the satisfaction of \emph{one-step} properties, e.g. decrease of a function across transitions, provided that transitions are indeed sampled independently.  
However, one is typically interested on inferring long (even infinite) horizon specifications from one-step properties. 
These approaches can reason only about finite horizon notions, as for infinite horizon properties the probability of satisfaction becomes trivially zero.

\textbf{Contributions.}
% \blue{omega or star language}
Addressing this limitation is the main objective of this work. 
% To simplify the presentation and provide a better insight on the basic idea 
We consider deterministic systems with unknown dynamics and uncertainty in their initialization. %(as done in most of the aforementioned works). 
Our approach provides a construction of data-driven finite abstractions, built on the notions of behavioural inclusions and $\ell$-complete behaviours. 
We introduce the notion of probabilistic behavioural inclusion, which is instrumental in describing the relation between a deterministic (but randomly sampled) model and a transition system constructed upon the collected system's behaviours. 
Leveraging the non-convex scenario theory, 
we establish PAC guarantees for the inclusion of the concrete system's (in)finite behaviours in those of the abstractions. 
This enables the verification of  infinite horizon properties while retaining the PAC guarantees.  
%
% We further assume that the initial state set and the whole state set coincide: this simplifies the acquisition of independent samples covering the whole state-set. 
%
We show throughout this work that the infinite horizon verification problem is equivalent to providing PAC guarantees on having observed the full support of an unknown probability mass function (with finite support) from a set of i.i.d. realizations. 
In other words, the PAC procedure bounds the probability of witnessing a new, unseen system behaviour. 

\textbf{Related Work.} 
The notion of system abstraction is generally related to a model-based technique, where an expert or an automated technique exploits the knowledge of a concrete system to extract some prominent features \cite{tabuada2009verification, baier2008principles} in order to verify some system's properties or to synthesise a policy (or controller) that ensures a user-defined requirement. 
When the model is unknown, one may resort to model identification (e.g. \cite{haesaert2017data}) to then verify the identified model via traditional formal techniques, in a two-step procedure. 
Recently however, a few works propose a one-step method, i.e. the use of PAC guarantees to directly synthesise an abstraction from data, with guarantees of correctness, without the need to identify an underlying model. 
In \cite{cubuktepe2020scenario, badings2021sampling, lavaei2022constructing} a sampled-based interval MDP is provided, employing the scenario approach to bound the transition probabilities of a stochastic dynamical model.
In \cite{devonport2021symbolic,lavaei2022data}, the authors define a PAC alternating simulation relationship between a symbolic abstraction and an underlying deterministic system, using one-step transitions. 

In \cite{makdesi2022data}, PAC over-approximations of monotone systems are computed, which are then used to build models for unknown monotone systems.
Finally,
\cite{kazemi2022data} computes the growth rate of a system from data, which is then used 
%with the sampled trajectories 
to construct a model abstraction and synthesise a controller.
The use of data-driven $\ell$-complete models is briefly presented in \cite{peruffo2022data} for linear PETC models. 
%
% barrier certificates
%
In \cite{salamati2021data, akella2022barrier} the authors synthesise barrier certificates for unknown systems using template-based candidates, providing PAC bounds for their correctness, for stochastic and  deterministic systems, respectively.
% In \cite{akella2022barrier}, unknown deterministic systems are considered. 
However, when PAC bounds are related to a system's transitions, the extension of behavioural guarantees from finite steps to infinite horizon is somewhat cumbersome, as we discuss in Section \ref{sec:traj-sample-abstract}.

%% verification

The scenario theory can be used also for verification purposes, as in 
\cite{akella2022scenario}, where the authors 
translate the probabilistic specifications into a Value-at-Risk determination problem for a scalar random variable. 
In \cite{wang2020data, wang2021scenario}, a scenario-based set invariance verification approach is proposed, which relies on the observation of trajectories of an unknown deterministic system. 

The works \cite{devonport2020estimating, devonport2021forwardsets, devonport2021reachchristoffel} discuss two approaches for the data-driven construction of  forward reachability set of unknown autonomous systems, in the form of $n$-dimensional spheres and using Christoffel functions.

%%%%%%%%%%%%%%%%%%%%%%%%%%%%%%%%%%%%%%%%%%%%%%%%%%%%%%%%%%%%%

\section{Preliminaries}
\label{sec:scenario}

\subsection{Notation}
The set of real numbers is denoted by $\mathbb{R}$.  We use a string notation for sequences, e.g. $\mathtt{r}=abc$ means $\mathtt{r}(1)=a$, $\mathtt{r}(2)=b$, $\mathtt{r}(3)=c$. When the length of a string $\mathtt{r}$ is not clear from the context, we use the subscript  $\ell\in\mathbb{N}_+$ to denote its length, i.e. $\mathtt{r}_\ell$. Given two sequences $\mathtt{r}_m$ and $\mathtt{p}_n$ with $\infty\geq m> n$, we say that $\mathtt{r}$ exhibits $\mathtt{p}$ if there exists $k\geq0$ such that $\mathtt{r}(k+i)=\mathtt{p}(i)$ for $i=\{1,...,n\}$, %and we denote this fact with the shorthand notation 
denoted $\mathtt{r}\models\Diamond \mathtt{p}$. Given a set of sequences $S$ and a sequence $\mathtt{p}$ we say that $S$ exhibits $\mathtt{p}$ if there exists $\mathtt{s}\in S$ such that $\mathtt{s}\models\Diamond \mathtt{p}$, denoted $S\models\Diamond \mathtt{p}$.  We denote the uniform distribution supported on a domain $\cD\subset\mathbb{R}^n$ by $\mathcal{U}_\cD$.

\subsection{Scenario Theory Background}
\label{subsec:scenario-background}

Let us offer an overview on the scenario theory, as outlined in \cite{campi2018general, garatti2021risk}, that we will use throughout this work. 
Let  $(\Delta,\mathcal{F},\mathbb{P})$ be a probability space, where $\Delta$ is the sample space, endowed with a $\sigma$-algebra $\mathcal{F}$ and a probability measure $\mathbb{P}$; further, we denote by $\Delta^N$ the $N$-Cartesian product of the sample space and with $\mathbb{P}^N$ its product measure.
% An element $\delta \in \Delta$ is interpreted as a potential situation to which the decision can be applied, while $\mathbb{P}$ describes the chance of such a situation to occur. 
A point in $(\Delta^N, \mathcal{F}^N, \bP^N)$
is thus a sample $(\delta_1 , \ldots , \delta_N )$ of $N$ elements drawn independently from $\Delta$ according to the same probability $\bP$.  
Each $\delta_i$ is regarded as an observation, or \emph{scenario}. 
A set $\Theta$, called the decision space, contains the decisions, i.e. the optimization space -- no particular structure is assumed for this set.  To every $\delta \in \Delta$ there is associated a constraint set $\Theta_\delta\subseteq\Theta$ which identifies the decisions that are
admissible for the situation represented by $\delta$.

Typically, the scenario theory refers to an optimisation program, which computes $\theta_N^*$, i.e. the solution of a optimisation program based on $N$ samples. 
Once $\theta^*_N$ is computed, we are interested in assessing how it generalises to unseen scenarios $\delta \in \Delta$, or, rather, the probability of extracting a sample that violates the constraints defined by $\theta^*_N$. To this end, we define the violation probability:
\begin{definition}[Violation \cite{campi2018general}] 
The violation probability of a given $\theta \in \Theta$ is defined as 
\begin{equation}
\label{eq:viol-prob}
    V ( \theta ) = 
    \mathbb{P} [ \, \delta \in \Delta \ | \ \theta \notin \Theta_\delta \, ]. 
\end{equation}
$V(\theta)$ quantifies the probability with which a new randomly selected constraint $\Theta_\delta$ is violated by $\theta$. If $V(\theta) \leq \epsilon$,  we say that $\theta$ is $\epsilon$-robust against constraint violation. 
\hfill $\square$
\end{definition}
Notice that in general $V (\theta)$ is not directly computable since $\mathbb{P}$ is not known.
In \cite{garatti2021risk}, under mild assumptions, the authors show that a confidence bound can be  derived as follows: 
\begin{theorem}[\cite{garatti2021risk}]
\label{theo:scenario-gurantees}
Given a confidence parameter $\beta \in (0, 1)$ and the solution $\theta^*_N$, it holds that 
\begin{equation}
\label{eq:scenario-confidence}
    \bP^N [
    V(\theta^*_N) \leq \epsilon(s^*_N, \beta, N)
    ] \geq 1 - \beta,
\end{equation}
where $\epsilon(\cdot)$ is the solution of a polynomial equation (omitted here for brevity) and $s^*_N$ is the so-called complexity of the solution --  it represents the minimum number of constraints ($m \leq N$) that yield the same solution $\theta^*_N$ computed with $N$ scenarios. 
\hfill $\square$
\end{theorem}
\begin{remark}
The canonical scenario theory assumes continuous spaces $\Delta$ and $\Theta$. 
We instead refer to non-convex scenario theory results, as the sample space $\Delta$ is discrete: this, in the scenario literature, is called a \emph{degenerate} problem. For a detailed discussion on the derivation of these results, the interested reader may refer to \cite{garatti2021risk, campi2018general}. 
\end{remark}

In this work, we solve the following problem.

\smallskip

\noindent\fbox{%
    \parbox{0.975\linewidth}{%
        \textbf{Problem Statement.} Given an unknown dynamical system, build an abstraction such that, with high confidence, the probability of witnessing a behaviour of the concrete system that is not exhibited by the abstraction's behaviours is below a threshold value.
    }%
}
%

%%%%%%%%%%%%%%%%%%%%%%%%%%%%%%%%%%%%%%%%%%%%%%%%%%%%%%%%%%%%%

\section{Sampling and Abstractions}
\label{sec:traj-sample-abstract}

\subsection{Modeling Framework}
\label{subsec:models}

Consider a time-invariant dynamical system described by
\begin{equation}
\label{eq:deterministic-sys}
    \Sigma(x):=
    \begin{cases}
        x_{k+1} = f( x_k ), 
        \\ 
        y_k = h(x_k),
        \\
        x_0 = x,
        \end{cases}
\end{equation}
where $x_k \in \mathcal{D} \subseteq \mathbb{R}^{n_x}$ is the plant’s state at time $k \in \mathbb{N}_+$, with initial value $x_0 \in \cD$,  $y_k \in \mathcal{Y}$ is the system output with $|\mathcal{Y}|<\infty$, and $n_x$ is the state-space dimension.
If the trajectory $x_k$ exits $\cD$ at time $k$, the output map returns a special symbol $y^{\dagger}$ for all $t \geq k$.
We may think of the map $h(\cdot)$ as a \emph{partitioning} map, that returns the partition label (or index) corresponding to any state $x_k$. The flow $f(\cdot)$ and output map $h(\cdot)$ are unknown, but we assume that given an initial condition $x_0$ we can observe the output sequence $y_0,y_1,...$ generated by $\Sigma(x_0)$. We denote $\cB^\omega(\Sigma)$ the set of infinite horizon behaviours (i.e. output trajectories) exhibited by the system $\Sigma$. With $\cB_H(\Sigma)$, we denote the set of behaviours for the time interval $k=[0, H-1]$. 
When the system of interest is evident, we remove the notation $\Sigma$. Let us now introduce the notion of %output sequence 
\emph{equivalence class}~\cite{tabuada2009verification}:
\begin{equation*}
    [ y ] = \{ x \in \cD\ 
    | 
    \ y = h(x) \},
\end{equation*}
and similarly, we define the equivalence class for an output sequence $\ty_{\ell_i}=y_{i_1}y_{i_2}...y_{i_\ell}\in\mathcal{Y}^\ell$ as
\begin{equation}
\label{eq:equivalence-class}
    [\mathtt{y}_{\ell_i}] = \{ x \in \cD\ 
    | 
    \ y_{i_j} = h(f^{j-1}(x)) \text{ for } j=1,...,\ell\},
\end{equation}
with $f^0(x)=x$.
Equation \eqref{eq:equivalence-class} states that 
for $i=1,...,|\mathcal{Y}|^\ell$, i.e. for every $\ell$-sequence $\ty_{\ell_i}\in\mathcal{Y}^\ell$ the output equivalence class $[\ty_{\ell_i}]$ is the set of points $x$ such that if the dynamical system is initialized at $x$, then the output sequence over the time interval $[0,\ell-1]$ corresponds to $\ty_{\ell_i}$. Note that an equivalence class might be empty if no initial condition gives an output behaviour corresponding to a particular $\ell$-sequence.
Further, for all $\ell \geq 1$, the set of all $[\ty_{\ell_i}]$ forms a partition of the domain $\cD$.

To compute the abstraction of an unknown dynamical system we assume that $\Sigma$ is randomly initialised by drawing initial conditions according to a probability measure space $(\cD,\mathcal{F}(\cD),\mathcal{P})$. 
Formally, we consider $\Sigma_s=(\Sigma(x_0),x_0\sim\mathcal{P})$ defined as
\begin{equation}
\label{eq:sys-probabilistic-init}
    \Sigma_s:=
    \begin{cases}
        \Sigma(x_0), 
        \\
        x_0 \sim \mathcal{P}(\cD).
        \end{cases}
\end{equation}
The addition of a probability measure over the initial states affects the behaviours of $\Sigma_s$. 
Let us refer to the probability of drawing an initial condition $x_0$ within the equivalence class $[\ty_{\ell_i}]$ as
\begin{equation*}
    \mathbb{P}[x_0 \in [\ty_{\ell_i}]] = \int_{[\ty_{\ell_i}]} p(v) dv, 
\end{equation*}
where $p(\cdot)$ represents the probability density function of $\mathcal{P}$.  
Throughout this work, abusing slightly notation, we consider $\cB_H(\Sigma_s)$
\begin{equation*}
    \cB_H(\Sigma_s) := \{ \ty_H \in \cB_H(\Sigma) \ | \  
    \mathbb{P} [ \,  x_0 \in [\ty_H] \, ] > 0
    \},
\end{equation*}
that is the set of behaviours associated to equivalence classes with strictly positive probability measure.
Obviously, it holds that if a behaviour belongs to $\cB_H(\Sigma_s)$ then it also belongs to $\cB_H(\Sigma)$, but the contrary is not necessarily true. 
%More precisely, the behaviours we collect depend on $\mathcal{P}$, as such behaviours have associated a probability mass. 
Some behaviours present on $\cB_H(\Sigma)$  may have a zero probability measure associated in $\Sigma_s$, therefore we say that $\Sigma_s$ almost surely does not exhibit them -- and we almost surely we do not sample them. For instance, an unstable equilibrium with an output associated exclusively to it $y_{eq}=h(x_{eq})$. In that case, while $\cB_H(\Sigma)$ contains the behaviour $(y_{eq})^H$, this behaviour is only exhibited if the singleton $\{x_{eq}\}$ is sampled, which almost surely does not happen. 
Note that, for the purpose of leveraging the scenario theory in this context, the presence of $\mathcal{P}$ is unavoidable. Abusing notation, we denote $\cB^\omega(\Sigma_s)$ and $\cB_H(\Sigma_s)$ the behaviours, with \emph{strictly positive} probability measure,
of infinite and finite time horizon $H$, respectively.
Notice that the PAC guarantees are derived for the behaviours of $\Sigma_s$, rather than for $\Sigma$.

%%%%%%%%%%%%%%%%%%%%%%%%%%%%%%%%%%%%%%%%%%%%%%%%%%%%%%%%%%%%%%%

\subsection{Abstractions via Transition Systems}
\label{subsec:abstractions}

The abstraction of a system is a tool that facilitates the analysis of large (even infinite) models. In this work, we assume that the concrete system is deterministic and non-blocking,
and
we employ a finite-state abstraction in the form of a transition system (TS). 
\begin{definition}[Transition System \cite{tabuada2009verification}]
    A transition system $\cS$ is a tuple $(\cX,\cX_0,\cE,\cY,\cH)$ where:
\begin{itemize}
    \item $\cX$ is the (possibly infinite) set of states,
    \item $\cX_0 \subseteq \cX$ is the set of initial states,
    \item $\cE \subseteq \cX \times \cX$ is the set of edges, or transitions, 
    \item $\cY$ is the set of outputs, and
    \item  $\cH : \cX \rightarrow \cY$ is the output map. 
    % \item $\gamma$ : $\cE \rightarrow \mathbb{Q}$ is the weight function.
\end{itemize}
\end{definition}
\smallskip
% We consider autonomous systems, and we discard the action set $\mathcal{U}$ of the original definition.
%
We consider \textit{non-blocking} transition systems, i.e. systems where every state is equipped with at least one outgoing transition.

Let us define $\mathtt{r} = x_0x_1x_2 \ldots$ an infinite internal behaviour, or \textit{run} of $\cS$ if $x_0 \in \cX_0$ and $(x_i, x_{i+1}) \in \cE$ for all $i \in \mathbb{N}$, and, with slight abuse of notation, $\mathcal{B}(\mathtt{r}) = y_0y_1y_2 \ldots$ its corresponding \textit{external} behaviour, or trace, if $\cH(x_i) = y_i$ for all $i \in \mathbb{N}$. With $\mathcal{B}^\omega(\cS(x))$ and $\mathcal{B}_H(\cS(x))$ we denote the sets of all infinite and of all finite external behaviours of length $H$  of $\cS$ starting from state $x$, respectively; when the system of interest is clear from the context we use the shorthand notation $\mathcal{B}^\omega(x)$ and $\mathcal{B}_H(x)$ respectively.
Models \eqref{eq:deterministic-sys} and \eqref{eq:sys-probabilistic-init} can be described as a transition system, as follows.
\begin{definition}[Embedding of a dynamical system]\label{def:embedding}
The embedding of a system $\Sigma$ is the transition system $\cS_\Sigma$ defined by $\cS_\Sigma=(\cD,\cD_0,\cE,\cY,\cH)$ where:
\begin{itemize}
    \item $\cE\subseteq \cD\times\cD$ is the set of all pairs $x_i,x_j$ s.t. $\cE = \{(x_i,x_j)\in\cD\times\cD:x_j=f(x_i)\}$,
    \item $\cH$ : $\cD \rightarrow \cY$ is the map defined by $h(x)$
\end{itemize}
Analogously if a system $\Sigma_s$ is endowed with a probability distribution over the initial states as in \eqref{eq:sys-probabilistic-init}, then the initial state is a random variable, and for every $x_0\in\cD_0$, it holds that $x_0\sim\mathcal{P}(\cD_0)$. 
\hfill $\square$
\end{definition}
Obviously, the embedding of a system does not change the set of all infinite (or finite) behaviours. Hence, the sets $\cB^\omega(\cS_{\Sigma_s})$ and $\cB^\omega(\cS_\Sigma)$ differ only by zero probability measure behaviours. When no ambiguity arises, we simply denote the embedding of $\Sigma_s$ as $\cS$, omitting the subscript.

%%%

In order to construct the embedding of \eqref{eq:sys-probabilistic-init}, we need full knowledge of the flow $f(\cdot)$ to compute the transitions $\cE$. 
As we only have access to the output $y_k$, 
we recall the notion of \emph{behavioural inclusion}: the abstraction is oblivious to the internal behaviour of a system, but we require that all behaviours that the concrete system can exhibit are within the set of behaviours of the abstraction. 
\begin{definition}[behavioural inclusion \cite{tabuada2009verification}]
\label{def:behave-inclus}
Consider two systems $\cS_a$ and $\cS_b$ with $\cY_a = \cY_b$. We say that $\cS_b$ behaviourally includes $\cS_a$, denoted by $\cS_a \preceq_\mathcal{B} \cS_b$, if $\mathcal{B}^\omega(\cS_a) \subseteq \mathcal{B}^\omega(\cS_b)$.
We say that 
$\cS_a$ is behaviourally included in $\cS_b$ until horizon $H$
if this holds until horizon $H$, i.e. 
$\cB_H(\cS_a) \subseteq \cB_H(\cS_b)$, 
denoted 
$\cS_a \preceq_{\cB_H} \cS_b$.
\hfill $\square$
\end{definition}
As we deal with systems where probability is embedded in the initial states, we define the probabilistic version of the behavioural inclusion relation. 
\begin{definition}[Probabilistic Behavioural Inclusion]
\label{def:prob-behave-incl}
Consider two systems $\cS_a$ and $\cS_b$ with $\cY_a = \cY_b$. We say that $\cS_a$ is behaviourally included in $\cS_b$ with probability greater or equal than $1-\epsilon$, denoted by 
$$\mathbb{P}[\cS_a \preceq_\mathcal{B} \cS_b]\geq 1-\epsilon, $$ 
if for $x_0\sim\mathcal{P}$ it holds that:
\begin{equation}
\label{eq:prob-behave-inclusion}
     \mathbb{P}\left[\cB^\omega(\cS_a(x_0))\subseteq \cB^\omega(\cS_b) \ | \ x_0 \sim \mathcal{P} \right]\geq 1-\epsilon,
\end{equation}
where $\cS_a(x_0)$ denotes the internal behaviour of system $\cS_a$ starting from $x_0$.
We say that $\cS_a$ is almost surely behaviourally included in $\cS_b$ if $\epsilon = 0$, written $\cS_a \preceq_\mathcal{B} \cS_b$ a.s. If the probabilistic behavioural inclusion holds until horizon $H$, analogously, we write
$$\mathbb{P}[\cS_a \preceq_{\mathcal{B}_H} \cS_b]\geq 1-\epsilon.$$
%\hfill $\square$
\end{definition}
A natural way of building a behavioural inclusion abstraction is by mapping each possible output sequence ${y}_i$ to an abstract state. We may elaborate this intuition through a so-called $\ell$-complete model:
\begin{definition}[(Strongest) $\ell$-complete abstraction \cite{schmuck2015comparing, de2021computing}]
\label{def:sl-ca}
Let $\cS :=(\cX,\cX_0,\cE,\cY,\cH)$ 
be a transition system, and 
let $\cX_\ell \subseteq \cY^\ell$ be the set of all $\ell$-long subsequences of all behaviours in $\cS$. 
Then, the system 
$\cS_\ell = (\cX_\ell, \cB_\ell(\cS), \cE_\ell, \cY^\ell, \cH)$
is called the (strongest) $\ell$-complete abstraction (S$\ell$-CA) of $\cS$, 
where
\begin{itemize}
    \item $\cE_\ell$ = $\{(k \sigma, \sigma k') \, | \, k,k' \in \cY$, $\sigma \in \cY^{\ell-1}$,  $k \sigma$, $\sigma k'$ $\in \cX_\ell$ \},
    \item $\cH(k \sigma) = k$,
\end{itemize}
where we denote $\cB_\ell(\cS)$ all the possible external traces of system $\cS_\ell$ and $\cY^\ell$ is the cartesian product $\cY \times \ldots \times \cY$ repeated $\ell$ times. 
%\hfill $\square$
\end{definition}

The intuition behind the S$\ell$-CA is to encode each state as an $\ell$-long external trace, as depicted in Fig.~\ref{fig:example_slca}.
Let us consider $\ell=3$ and  
assume a three-step trajectory $x_0 x_1 x_2$ that returns an output $y_0 y_1 y_2$. %i.e. $x_k \in [y_0 y_1 y_2]$.
This trajectory is mapped into the abstract state $\mathtt{y}_\ell$ corresponding to the 3-sequence $\mathtt{y}_\ell = y_0 y_1 y_2 $.
Notice that the transitions of an $\ell$-complete model follow the so-called ``domino rule'': e.g., let us assume $\ell=3$ and let us observe a trace $y_0 y_1 y_2 $; 
the next $\ell$-trace must begin with $y_1 y_2 $. 
Therefore, the state $y_0 y_1 y_2 $  can transition to 
e.g. $y_1 y_2 y_0 $, $y_1 y_2 y_1 $, $y_1 y_2 y_2 $.
Finally, the output of a state is its first element: state $y_0 y_1 y_2 $ has output $\cH(y_0 y_1 y_2 ) = y_0$.

The $\ell$-complete model does not require the knowledge of the concrete flow $f(\cdot)$ to construct its transitions, as it only relies on the external behaviours (i.e. the output) of a model to construct its states. 
The value $\ell$ defines the trajectory horizon and practically embeds the future of each concrete state $x_k$.
The S$\ell$-CA generates all possible behaviours of an underlying system, as long as its state space (i.e. the $\ell$-sequences) contains all the $\ell$-trajectories that the system may exhibit. 
We only need the possible $\ell$-behaviours of a system to construct an $\ell$-complete model.

\begin{figure}
    \centering
    \input{models/example-slca}
    \caption{Example of S$\ell$-CA, with $\ell=2$ (left) and $\ell=3$ (right).}
    \label{fig:example_slca}
\end{figure}
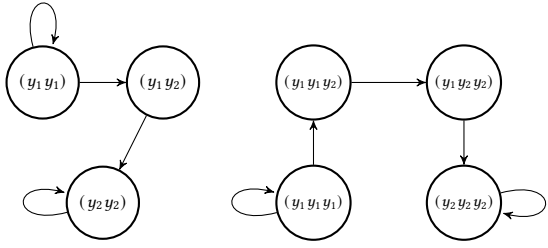

%%%%%%%%%%%%%%%%%%%%%%%%%%%%%%%%%%%%%%%%%%%%%%%%%%%%%%%%%%%%%%%%%%

\subsection{Trajectory Sampling}
\label{subsec:traj-sampling}

Several approaches \cite{badings2021sampling, cubuktepe2020scenario, devonport2021symbolic,  lavaei2022constructing, makdesi2022data,lavaei2022data}
% collect one-step transitions and 
provide PAC guarantees based on the concrete system's transitions: in particular, samples are used to compute  
either transition probabilities or the probability of missing a transition amongst abstract states. 
As probabilities are related to \emph{transitions}, the analysis of infinite-horizon properties becomes cumbersome: at every step, the transition probabilities accumulate until the abstraction becomes uninformative.

Further, 
we need to carefully address the extension of guarantees stemming from one-step transitions into multiple steps.
The samples belonging to a trajectory are inherently correlated and they belong to different probability distributions. 
In practice, we sample the initial state $x_0$ according to a probability distribution $\mathcal{P}$, i.e. 
\begin{equation*}
    x_0 \sim \mathcal{P},  
\end{equation*}
and we then compute $x_1$, which may arise from a deterministic or stochastic system. 
Recall that  
the scenario theory can be used to 
provide guarantees about the behaviours arising from samples \emph{extracted from the same distribution}.
In  general, however,
the value 
\begin{equation*}
    x_1 = f(x_0) + \eta_0 \nsim \mathcal{P}, 
\end{equation*}
belongs to a \emph{different} probability distribution,  as $f(x_0)$ belongs to a $\mathcal{P}$ distribution transformed by $f(\cdot)$ and the noise $\eta_0$ -- if present -- may belong to any other distribution; overall, $x_0$ and $x_1$, belong to \emph{different} distributions. 
Hence, although one may use one-step transitions to solve a scenario program, i.e. 
$$
\{\delta_i\}_{i=1}^N = \{(x_0 \rightarrow x_1)_i\}_{i=1}^N,
$$ 
the scenario theory provides guarantees for a new one-step transition $(x_0 \rightarrow x_1)_{N+1}$. Nothing can be said on $(x_1 \rightarrow x_2)_{N+1}$, without further assumptions, as $x_1\nsim\mathcal{P}$.

The following example shows the shortcomings of using one-step transitions to infer (in)finite-horizon properties. 
\begin{example} 
Let us assume to partition the state space via a gridding procedure, as depicted in Fig. \ref{fig:scenario-scheme}, and let us sample initial conditions from $P_0$, the partition in the bottom left corner. 
All trajectories starting from $P_0$ reach the blue region of  $P_{4}$, where the blue region is strictly smaller than $P_4$.
Let us further assume that every trajectory starting from the blue region reaches an unsafe region ($P_6$) in one step, whereas the trajectories starting from the white portion reach a safe set ($P_7$) in one step. 
If we sample trajectories from the partition $P_{4}$, we see that \emph{some} of these reach the safe set (the ones starting from the white portion) and others reach the unsafe set (starting from the blue portion). 
Whilst the probability of reaching the unsafe set from $P_0$ is 1, using the sampled one-step transitions yield a probability  of reaching the unsafe set strictly smaller than 1 and actually suggests the safe set can be reached. 
\hfill $\square$
\end{example}

\begin{figure}
    \centering
    \input{models/grid-part}
    \caption{Abstraction based on one-step transitions.}
    \label{fig:scenario-scheme}
\end{figure}
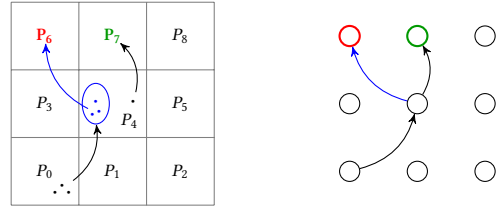

Our approach is crucially different from the related literature: 
we employ the scenario approach to evaluate the probability of (existence of) \emph{states}, whilst we use a deterministic rule (i.e. the domino rule) for transitions. 
% 

%%%%%%%%%%%%%%%%%%%%%%%%%%%%%%%%%%%%%%%%%%%%%%%%%%%%%%%%%%%%%%%%%%%%%%%%%%%%%%%%%

%
\subsection{Data-driven Abstraction: Finite-time Guarantees}
\label{subsec:behav-finite-time}

To overcome the limitations stemming from one-step transitions, we collect trajectories up to time $H$ and reason about the system behaviour over such horizon.
Each $H$-long trajectory can be divided into several $\ell$-sequences (or one, if $\ell = H$). For instance, consider $\ell=2$, $H=4$, and assume we sample the trajectory $y_1y_2 y_3 y_4$;  the trajectory is split into $(H-\ell+1)$ 2-sequences: $y_1 y_2$, $y_2 y_3$, $y_3 y_4$, and each $\ell$-sequence corresponds to an  abstract state of the \slca.

The data-driven procedure works as follows. 
We collect $N$ $H$-long trajectories and construct $\cX_\ell^N$, the set of witnessed $\ell$-sequences $\ty_\ell$, that acts as the state set of the data-driven S$\ell$-CA. 
\begin{definition}[Data-driven S$\ell$-CA]
\label{def:data-driven-slca}
The $\ell$-complete abstraction  
$\cS_\ell^N = (\cX_\ell^N, \cX_\ell^N, \cE_\ell, \cY^\ell, \cH)$ 
is called the data-driven $\ell$-complete abstraction (S$\ell$-CA) of $\cS$, 
where
\begin{itemize}
    \item $\cX^N_\ell$ is the state space built from the $\ell$-sequences collected from  $N$ trajectories of an underlying concrete system.
\end{itemize}
The transitions, output space, and output map follow Definition \ref{def:sl-ca}.
\hfill $\square$
\end{definition}

\begin{remark}[Domino Completion]
The structure of the domino transitions, coupled with the sample-based states, may give rise to a \emph{blocking} transition system. 
To illustrate this issue,
consider $\ell=3$ and a dataset that exhibits solely the sequences $\{ y_1 y_1 y_1,\, y_1 y_1 y_2,\, y_1 y_2 y_1 \}$. As the dataset does not include any sequence starting with $y_2 y_1$, the state corresponding to $y_1y_2y_1$ has no outgoing transitions, as depicted in Fig. \ref{fig:blocking-TS}.
On the other hand, the existence of the sequence $y_1y_2y_1$ \emph{implies} the existence of at least one sequence starting with $y_2y_1$. 
To overcome this impediment, we artificially add \emph{all} states corresponding to sequences $y_2 y_1 *$.
We repeat the procedure until we obtain a non-blocking transition system.
% Less conservative procedures are also possible, e.g. adding a minimal set of states that complete the transitions: in this example, we may add only state $y_2y_1y_1$.
\hfill $\square$
\end{remark}

\begin{figure}
    \centering
    \input{models/blocking}
    \caption{Construction of a non-blocking automaton. Dashed lines indicate artificial states and transitions, added by the domino completion.}
    \label{fig:blocking-TS}
\end{figure}
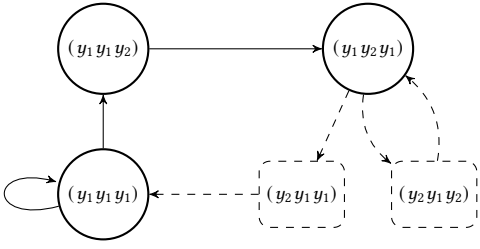

\smallskip

From here on we assume that the data-driven S$\ell$-CA is non-blocking. Let us recall that we sample trajectories from $\Sigma_s$ in \eqref{eq:sys-probabilistic-init}: the set $\cX_\ell^N$ almost surely does not contain $\ell$-sequences related to zero-measure sets. 
Once we collect $N$ trajectories from the concrete system and construct the corresponding data-driven \slca, we leverage the scenario theory to provide bounds on the probabilistic behavioural inclusion between the data-driven abstraction $\cS_\ell^N$ and the embedding of the concrete system $\cS$. In practice, we upper bound the probability of drawing an initial condition $x_0$ in the concrete system that produces an $H$-long behaviour which cannot be generated by the abstraction.
%
% As we will show, we upper bound the cumulative probability mass of the unseen support of a finite probability mass function from a set of i.i.d. realizations.

As a first step towards the definition of the scenario problem, let us outline a chance-constrained program, where we assume to know both the probability distribution $\mathcal{P}$ and the concrete system $\Sigma$. We denote by $\ty_{\ell_ i}$ the $i$-th $\ell$-sequence of the set $\cY^\ell$, for $i\in\{1,...,|\cY|^\ell\}$. 
For the sake of simplicity, let us formulate the program for $\ell=H$ (if $H>\ell$ the same reasoning applies), and $H \geq 1$. Let us define the random vector $\delta\in\Delta$ as follows:
\begin{multline}
\label{eq:delta-one-hot-vector}
    \delta = [X_1,X_2,...,X_{|\cY|^\ell}]^T, 
    \text{ where }
    \\
% \end{equation}
% where
% \begin{equation*}
    X_i  \sim\text{Bernoulli}(p_i),
    \qquad
    p_i  = \mathbb{P}[ \,
    x_0: \cB_H(x_0) \models \Diamond \ty_{\ell_ i}
    \, ],
\end{multline}
for $i\in\{1,...,|\cY|^\ell\}$ and $x_0\sim\mathcal{P}$. 
Note that, since the flow $f(\cdot)$ is deterministic, the $\delta$ depends entirely on $x_0$, hence the $X_i$'s are mutually dependent. 
The vector $\delta$ has one unique entry equal to 1 when $\ell= H$, hence $\delta$ follows a categorical distribution; otherwise it may have multiple ones.

\begin{example}
\label{ex:chance-constrained}
Let us assume a model has output space $\cY = \{a, b, c, d, e\}$, and let us consider $\ell=H=2$, thus giving $|\cY|^\ell = 25$. 
Let us assume that the underlying system produces solely the 2-sequences $aa$, $bc$, $cd$, $de$, $eb$, whilst it does not admit any other $2$-sequence. 
%
% Let us further assume that the probability of sampling $aa$ is $0.04$, whilst it is $0.24$ for $bc$, $cd$, $de$, $eb$.  
The concrete system partitions the domain into 5 regions, as depicted in Fig.\ref{fig:chance-constr}, which correspond to the equivalence classes $[aa]$, $[bc]$, $[cd]$, $[de]$, $[eb]$. 
The vector $\delta$ is composed of 25 elements; and 
$p_1$, i.e. the probability of sampling $aa$, is equal to $0.04$. The probabilities of $bc$, $\ldots$, $eb$ are respectively $p_2, \ldots, p_5$, which are equal to $0.24$. The remaining $p_i$ for $i>5$ are zero.
At each initial condition $x_0$, randomly sampled within the domain, corresponds one $\ell$-sequence and one vector $\delta$. For instance, 
sampling the 2-sequence $aa$ has probability 0.04 and is encoded as the vector $\delta = [1, 0, \ldots, 0 ]$. 
Notice that \emph{knowing} that there are only 5 possible sequences reduces the support of $\delta$; however, when this knowledge is unavailable, we use $|\cY|^\ell$. 
\hfill $\square$
\end{example}

\begin{figure}
    \centering
    \input{models/histogram}
    \caption{Domain partitions and probability of sampling the $2$-sequences.}
    \label{fig:chance-constr}
\end{figure}
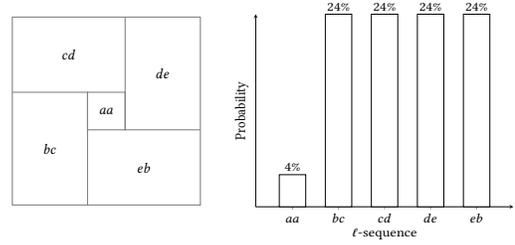

\begin{remark}
For the purpose of obtaining an exact S$\ell$-CA of $\cS$ such that $\cS\preceq_{\mathcal{B}} \cS_{\ell}$, 
we solely require the knowledge about the possible (i.e. realisable) $\ell$-sequences; 
in practice it would be sufficient to know for which $i\in\{1, \ldots ,|\cY|^\ell\}$ it holds that $p_i>0$. In that case, the relation $\cS \preceq_{\mathcal{B}} \cS_{\ell}$ holds almost surely (see Definition \ref{def:prob-behave-incl}).
\hfill $\square$
\end{remark}

As $\Delta$ is a discrete set with cardinality not greater than $|\cY|^\ell$, we can rewrite the $i$-th $\ell$-sequence $\ty_{\ell_i}$ as a one-hot vector $\delta_i$ of length $|\cY|^\ell$, as outlined in \eqref{eq:delta-one-hot-vector}, and let $\Theta~=~\mathbb{R}^{|\cY|^\ell}$.
%
% Then, the set $\Delta$ can have \emph{at most} $|\cY|^\ell$ elements, i.e. $\Delta = \{\delta^{(1)},\delta^{(2)},...,\delta^{(|\cY^\ell|)}\}$. 

% We are now ready to state the chance-constrained program that we aim to solve \red{Here we actually aim to solve it, but can't though.}:
We are now ready to state the chance-constrained problem:  
%which acts as the foundation of a scenario program:
\begin{equation}
\label{eq:chance-constrained-prog}
\begin{aligned}
    &
    \min_{\theta \in \Theta, \Delta_\epsilon}
    &
        \mathbf{1}_{|\cY|^\ell}^{\text{T}} \cdot \theta &, 
        \\
        &
        \text{s.t. }
        &
        \quad
        \sum_{i=1}^{|\cY|^\ell} p_i \cdot  \mathbbm{1}_{\Delta_\epsilon}(\delta_i) \geq 1-\epsilon ,
        &
        \\
        &
        &
        \quad
        ( \theta-\delta_i) \geq 0, 
        & \text{ for } \delta_i\in\Delta_\epsilon
\end{aligned}
\end{equation}
where $\Delta_\epsilon\subseteq\Delta$, with $\mathbb{P}(\Delta_\epsilon)\geq1-\epsilon$, $\mathbf{1}_{|\cY|^\ell}$ is a column vector of ones with length $|\cY|^\ell$, $\mathbbm{1}_{\Delta_\epsilon}(\cdot)$ is the indicator function.
%
% Given an optimal solution $(\theta^*,\Delta_\epsilon^*)$, we can construct an approximate S$\ell$-CA of $\cS$, denoted by $\cS_{\ell}^\epsilon$. 
%
The program can be interpreted as follows. Given a threshold $\epsilon$, find the maximum number of $\delta_i$ such that the sum of their probabilities is smaller than $\epsilon$.  
Recalling Example \ref{ex:chance-constrained}, given $\epsilon = 0.05$, we shall find a set of events $\Delta_\epsilon$ such that its probability mass is not smaller than $1-\epsilon$. 
From the event space $\Delta = \{aa, bc, \ldots eb\}$ we then discard $aa$, since its probability is 0.04. 
We can then build an abstraction composed of the four remaining $\ell$-sequences, which accounts for a cumulative probability $\mathbb{P}(\Delta_{\epsilon}) \geq 1 - \epsilon$. 

The optimal solution of \eqref{eq:chance-constrained-prog} is denoted $(\theta^*, \Delta_\epsilon^*)$, where  $\theta^*$ represents a vector encoding the $\ell$-sequences that satisfy the probabilistic constraint  $\mathbb{P}(\Delta_\epsilon)\geq1-\epsilon$. We can thus construct an approximate S$\ell$-CA, where the states derive from $\theta^*$, and the transitions are governed by the domino rule.
In fact, if the $i$-th row of $\theta^*$ is non-zero ($\theta^*\in\{0,1\}^{|\cY|^\ell}$), then $p_i>0$, and $\ty_{\ell_i}$ is part of the state set of $\cS_\ell^\epsilon$. Moreover, the domino completion might be necessary if the approximate S$\ell$-CA has blocking states.  Note that there could be more than one optimal solution, since $\Delta$ is a discrete (and finite) set -- in that case, we choose the smallest $\theta^*$ in lexicographic order. 
The following result follows trivially:
\begin{proposition}\label{prop:chance-constrained-prog}
Given $\epsilon$, for any optimal solution $(\theta^*,\Delta_\epsilon^*)$ to (\ref{eq:chance-constrained-prog}) and the corresponding approximate S$\ell$-CA, $\cS_{\ell}^\epsilon$, for a new initial condition $x_0$ sampled from $\mathcal{P}$ it holds that
\begin{equation}
\label{eq:chance-constrained-simulation}
    \mathbb{P}\left[\cB_H(\cS(x_0))\in \cB_H(\cS_{\ell}^\epsilon)\right]\geq 1-\epsilon,
\end{equation}
where $\cB_H(\cS(x_0))$ denotes the $H$-long behaviour exhibited by $\cS$, i.e. the embedding of $\Sigma_s$, starting from $x_0$, and $\cB_H(\cS_{\ell}^\epsilon)$ denotes the set of all $H$-long behaviours of $\mathcal{S_\ell^\epsilon}$. 
\end{proposition}
Proposition \ref{prop:chance-constrained-prog} states that the set of $H$-long behaviours which belong to $\cS$ but do not belong to $\cS_\ell^\epsilon$ have a probability measure not bigger than $\epsilon$. 
As outlined in Definition \ref{def:prob-behave-incl},
we can rewrite \eqref{eq:chance-constrained-simulation} as
\begin{equation*}
    \mathbb{P}[\mathcal{S}\preceq_{\cB_H} \mathcal{S}_\ell^\epsilon]\geq 1-\epsilon.
\end{equation*}
Since the flow $f(\cdot)$ is assumed to be completely unknown, it is not possible to solve (\ref{eq:chance-constrained-prog}) exactly. We resort to scenario theory to find an approximate solution to (\ref{eq:chance-constrained-prog}): as both $\theta$ and $\delta$ take values over a discrete set, we refer to the general scenario theory \cite{garatti2021risk}. 

\smallskip

Let us sample $N$ i.i.d. initial conditions $\{x_{0,i}\}_{i=1}^N$ in the dynamical system, and consider the resulting $H$-long behaviours displayed by $\cS$, denoted by $\{\cB_H(x_{0,i})\}_{i=1}^N$. We directly obtain $N$ i.i.d scenarios $\{\delta_i\}_{i=1}^N$ where 
\begin{equation*}
    \delta_i(j) =  
    \begin{cases}
    1 \text{ if } \cB_H(x_{0, i})\models\Diamond\mathtt{y}_{\ell_j}, \\
    0 \text{ else },
    \end{cases}
\end{equation*}
for $j\in\{1,...,|\cY|^\ell\}$. 
% Then it is possible to set a scenario program for which the optimal solution, obtained from the scenarios $\{\delta_i\}_{i=1}^N$, is violated by a new scenario $\delta_{N+1}$, given by the initial condition $x_{0,N+1}$, if and only if at least one more new $\ell$-sequence is observed in $\cB_H(x_{0,N+1})$. 
We define formally the scenario program as
\begin{equation}
\label{eq:scenario-worst-case-problem}
    \begin{aligned}
        &
        \min_{\theta \in \Theta} 
        & \mathbf{1}_{|\cY|^\ell}^{\text{T}} \cdot \theta
        &
        \\
        &
        s.t. 
        &
        ( \theta-\delta_i) \geq 0, 
        & 
        \quad i=1, \ldots, N.
    \end{aligned}
\end{equation}
The solution $\theta^*_N$ is trivially unique and  in practice indicates which $\ell$-sequences were witnessed in the samples collected; the solution changes solely when we collect a new value for $\delta_i$, previously unseen. 
Then, if $\ell=H$, the complexity  $s^*_N$ is equal to the number of 1's in the vector $\theta^*_N$, or in other words, $s^*_N$ is equal to the number of different $\ell$-sequences exhibited by the $N$ $H$-sequences. If $H>\ell$, the complexity is equal to the cardinality of the smallest subset of the $N$ $H$-sequences collected that yield the same solution to $\theta^*_N$. 

\begin{remark}
\label{rem:support-size}
We interpret program \eqref{eq:scenario-worst-case-problem} as the collection of labels (i.e. the $\ell$-sequences of partitions of the concrete model) from a discrete probability distribution of unknown support size. 
The scenario theory provides a bound to the probability of collecting a new, unseen, label from the unknown distribution. In this sense, we upper bound the cumulative probability mass of the unseen support of a finite probability mass function from a set of i.i.d. realizations.
As a second step, we employ the collected labels to construct a data-driven abstraction, which acquires the scenario probability guarantees. 
%as the following proposition states. 
%
% If we collect \emph{all}  the $\ell$-sequences that the system can exhibit (almost surely), then an abstraction  $\cS_\ell^N$ constructed as per Definition  \ref{def:data-driven-slca} almost surely behaviourally includes the underlying model, by construction. However, we do not know the exact number of $\ell$-sequences with positive probability: we thus exploit the scenario theory to bound the probability of witnessing a new $\ell$-sequence. 
\hfill $\square$
\end{remark}
We equip a data-driven \slca \ with PAC guarantees as follows. 
\begin{proposition}
\label{prop:behav-incl-horizon}
Consider a confidence $\beta$, and  
    $N$ trajectories of length $H$ collected from \eqref{eq:sys-probabilistic-init},
    and 
    the corresponding data-driven S$\ell$-CA  $\cS^N_\ell$ based on the observed $\ell$-sequences. 
    For a new initial condition $x_0$ sampled from $\cD$ it holds that
    \begin{equation}
    \label{eq:scenario-simulation-time-H}
        \mathbb{P}^N [
    \mathbb{P}[\cB_H(\cS_{\Sigma_s}(x_0))\in \cB_H(\cS^N_\ell)
    ]
     \geq 1 - \epsilon(s^*_N, N, \beta)
     \,
    ]
    \geq 1 - \beta, 
    \end{equation}
    where $\cB_H(\cS_{\Sigma_s}(x_0))$ denotes the $H$-long behaviour exhibited by $\mathcal{S}_{\Sigma_s}$ starting from $x_0$ and $\cB_H(\cS_{\ell}^N)$ denotes the set of all $H$-long behaviours of $\mathcal{S}_\ell^N$.
    % \hfill $\square$
\end{proposition}
\begin{proof} 
By definition of an S$\ell$-CA, the states (i.e. the $\ell$-sequences) are the one element related to the concrete system. 
As the domino transitions are defined deterministically, we only have to ensure to collect all possible states from data. 
% If we collect all states, the behavioural relation holds; on the other hand, the scenario theory bounds the probability of unseen sequences by $\epsilon$, hence \eqref{eq:scenario-simulation-time-H}. 
%
If we collect all the possible $\ell$-sequences that the system can exhibit, the behavioural relation holds, as the resulting abstraction simulates the concrete system (see e.g. \cite{de2021computing}). 
Let us now assume we collect only a subset of $\cB(\Sigma)$. 
The scenario theory guarantees that the probability of sampling an initial condition $x_0$ that generates an unseen $\ell$-sequence, over the time horizon $H$, is bounded by $\epsilon$; hence, \eqref{eq:scenario-simulation-time-H} holds.
Recall that if we sample a set of states that form a blocking TS, we apply the domino completion: this improves the behavioural inclusion, as it increases the number of behaviours. 
\hfill $\square$
\end{proof}

\begin{figure}[h]
    \centering
    \includegraphics[width=0.85\linewidth]{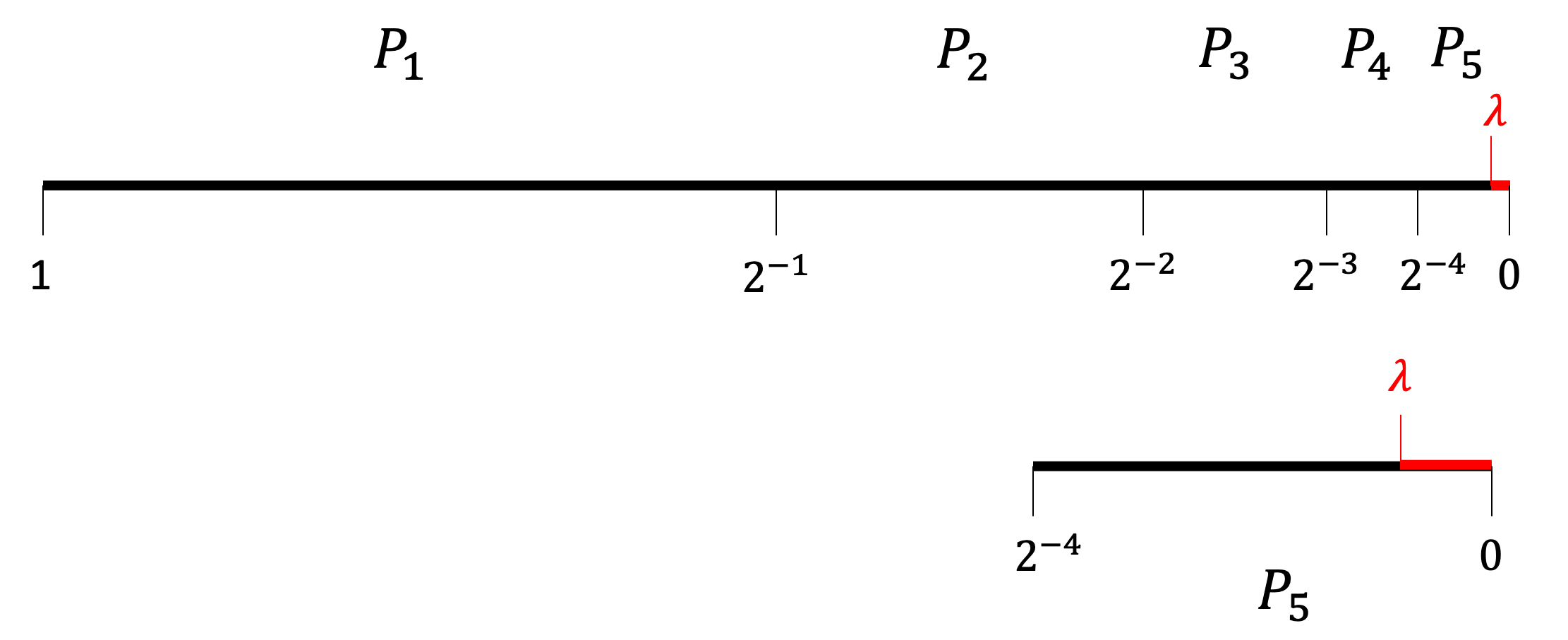}
    \caption{Partition for the state space of the dynamical system in Example \ref{exa:ell-seq-infinity}.}
    \label{fig:contracting-sys}
\end{figure}
\smallskip

As remarked in Proposition \ref{prop:chance-constrained-prog}, we denote this property by
\begin{equation}
\label{eq:data-driven-simulation-relation}
    \mathbb{P}^N [
    \mathbb{P}[\mathcal{S}\preceq_{\cB_H} \mathcal{S}_\ell^N
    ]
    \geq 1 - \epsilon(s^*_N, N, \beta)
    \,
    ]
    \geq 1 - \beta.
\end{equation}
Let us present a simple example that shows the importance of the finiteness of the horizon $H$ in Proposition \ref{prop:behav-incl-horizon}, and why the extension of \eqref{eq:data-driven-simulation-relation} to infinite horizons is challenging.

\begin{example}
\label{exa:ell-seq-infinity}
The scenario theory bounds the probability of
sampling a new, 
% unseen $\ell$-sequence as the initial portion of an infinite trajectory.
unseen $\ell$-sequence within horizon $H$.
Let us consider a one dimensional system, 
\begin{equation}
\label{eq:contracting-sys}
    x_{k+1}=
    \begin{cases}
    \frac{1}{2}x_k, &\text{if } x_k \in (\lambda, 1]
    \\
    \frac{1}{2}x_k + \frac{1}{2} &\text{if } x_k \in [0, \lambda]
    \end{cases}
\end{equation}
where $x_0\sim\mathcal{U}_{[0,1]}$, and $0<\lambda<2^{-4}$. 

The state space $[0, 1]$ is partitioned into five regions according to $P_i = ( 2^{-i}, 2^{-i+1}]$ for $i=1,...,4$, and $P_5 = [0, 2^{-4}]$, as shown in Fig. \ref{fig:contracting-sys}. Let us denote by $y_i$ the output of the system if the state belongs to $P_i$.
It is easy to see that this system visits infinitely many times all 5 partitions, no matter what the initial condition is. For instance, if $x_0 \in P_1$, the system generates the repeating sequence $(y_1y_2y_3y_4y_5^t)^{\omega}$, where $t$ denotes the number of repetitions of the output $y_5$ and depends on $\lambda$, i.e. the width of the window that makes the system jump back to $P_1$.

% The following discussion remarks %\blue{how the choice of the horizon $H$ influences the amount of collected  $\ell$-sequences. }
% the fundamental assumption necessary for the scenario theory, i.e. the random process yielding the scenarios used to compute the solution to a scenario program must remain the same for the guarantees to hold.
Let us now assume we sample uniformly $[0,1]$ and take $\ell=H=2$:  this is equivalent to considering one-step transitions. It is easy to see that the probability of witnessing any $y_iy_{i+1}$ is equal to $2^{-i}$ for $i = 1,...,4$. 
The probability of sampling $y_5 y_5$ is $2^{-4} - \lambda$, and 
% Moreover, for $\lambda \ll 0.5^4$, the probability of witnessing $\ty_5\ty_5$ is approximately $0.5^4$. On the other hand, we may 
we witness the sequence $y_5y_1$ only if $x_0 \in (0,\lambda]\subset P_5$, an event that occurs with probability $\lambda$, as summarized below: 
\begin{equation}\label{eq:contracting-sys-horizon-2}
\begin{aligned}
    &
    \mathbb{P}[\cB_2(x_0) \models \Diamond y_1 y_2] = 2^{-1},
    &
    \mathbb{P}[\cB_2(x_0) \models \Diamond y_2 y_3] = 2^{-2},
    \\
    &
    \mathbb{P}[\cB_2(x_0) \models \Diamond y_3 y_4] = 2^{-3},
    &
    \mathbb{P}[\cB_2(x_0) \models \Diamond y_4 y_5] = 2^{-4},
    \\
    &
    \mathbb{P}[\cB_2(x_0) \models \Diamond y_5 y_5] = 2^{-4} - \lambda,
    &
    \mathbb{P}[\cB_2(x_0) \models \Diamond y_5 y_1] = \lambda, 
    \\
    &
    \mathbb{P}[\cB_2(x_0) \models \Diamond y_5 y_2] = 0, 
\end{aligned}
\end{equation}
Note that $\lambda$ is a parameter of the system and can be arbitrarily small.

Let us now consider a longer horizon $H' > 2$ and focus on the sequence $y_5 y_1$. 
%  Clearly, claiming that the probability of witnessing the sequence $y_5y_1$ is $\lambda$ also for simulations with horizon $H'>2$ is incorrect. In fact, 
It is easy to see that the subsequence $y_5 y_1$ will \emph{eventually} be generated by \emph{every} trajectory, if the horizon $H'$ is long enough\footnote{By inspection, the maximum number of steps to see $y_5 y_1$ is $1+\left \lceil{-\log_2(\lambda)} \right \rceil$.}. 
Formally, for $H'\geq 1-\left \lceil{\log_2(\lambda)} \right \rceil$ it holds that 
\begin{equation*}
    \mathbb{P}[\cB_{H'}(x_0)\models \Diamond y_5y_1] = 1.
\end{equation*}
This represents a challenge for our approach: whilst the probability of witnessing $y_5 y_1$ \emph{as an initial sequence} is arbitrarily small, the probability of seeing $y_5 y_1$ over a sufficiently long horizon is actually 1. 
In general, we cannot use the scenario bounds \emph{generated from samples over time horizon $H$} to infer properties over longer time horizons. 
Finally, observe that the system exhibits the behaviour $y_5 y_2$ if and only if it is initialized exactly at $x_0 = 0$. Such event has a zero probability measure since $x_0\sim\mathcal{U}_{[0,1]}$, hence, any data-driven S$\ell$-CA will almost surely not include such behaviour.
\hfill $\square$
\end{example}

\subsection{Trajectory Horizon and Choice of $\ell$}\label{sec:tuning-parameters}

Let us offer some intuition on the role of $\ell$ and $H$.
On the one hand, we are interested in considering a large $H$ as we aim to 
% ensure that the  behavioural inclusion holds for 
extend the PAC guarantees over 
a long time horizon. 
In addition, if $H \gg \ell$, 
we expect to obtain a data-driven \slca \ containing a large portion of all the possible $\ell$-sequences, as every trajectory contains several $\ty_\ell$'s.
Further, 
% the sample space has cardinality not greater than $|\cY|^\ell$, hence 
a small $\ell$ limits the number of possible sequences (upper bounded by $|\cY|^\ell$), which results in more compact abstractions. 
On the other hand, 
$\ell$ can be considered a refinement parameter: a larger $\ell$ provides a finer and more precise  abstraction.
Increasing the state space typically reduces the nondeterminism of the abstraction. 
Notice that this refinement applies also in case of a black-box model, where we do not have access to the partitioning map. 
The ``classical" state discretization, i.e. gridding procedure, needs access to the actual state variables, not only to the output ones, to refine the partitioning. 

We highlight that neither $\ell$ nor $H$ appear explicitly in the computation of $\epsilon$, however, they affect the complexity $s^*_N$. 
Recall the definition of $s^*_N$ from Theorem \ref{theo:scenario-gurantees}: for a given set of $N$ scenarios $\{\delta_i\}_{i=1}^N$ and the corresponding optimal solution $\theta^*_N$, the complexity $s^*_N$ is equal to the cardinality of the smallest subset $\{\delta_{i_1},...\delta_{i_m}\}\subseteq\{\delta_i\}_{i=1}^N$ returning the same scenario program solution, that is, $\theta^*_N=\theta^*_m$ and $s^*_N = m$. The choice of $H$ is directly linked to the duration allowed for \emph{exploration} of the state space of the system. In fact, every $H$-trajectory contains between $1$ and $\min(|\cY|^\ell, H -\ell + 1)$ different $\ell$-sequences. It follows that, as we increase the gap between $H$ and $\ell$, it is likely that the complexity will decrease, since one $H$-sequence can include multiple $\ell$-sequences. 
For a fixed $N$, the lower the complexity $s^*_N$, the lower the $\epsilon (s^*_N,N,\beta)$, and therefore, the higher is the probability of behavioural inclusion. 

Notice that finding the complexity of \eqref{eq:scenario-worst-case-problem} is equivalent to solving a set cover problem \cite{caprara2000algorithms}, where the universe is given by all the witnessed $\ell$-sequences, and every covering set is obtained by considering the unique $\ell$ subsequences present in every sampled $H$-sequence. %the covering sets \blue{derive from} %relate the $N$ $H$-sequences.
In practice, we use a greedy algorithm for the (unweighted) set cover problem to find an upper bound of the complexity (and hence an upper bound for $\epsilon$), see \cite{young2008greedy}. We construct the minimum cardinality of H-sequences $\{\delta_{i_1},...,\delta_{i_m}\}$ holding the same solution $\theta^*_N$, by adding iteratively the $H$-sequence containing the largest number of the witnessed $\ell$-sequences.

%%%%%%%%%%%%%%%%%%%%%%%%%%%%%%%%%%%%%%%%%%%%%%%%%%%%%%%%%%%%%

\section{Infinite Behaviours}
\label{sec:infinite-behav}

Proposition \ref{prop:behav-incl-horizon} states that we can construct an abstraction that behaviourally includes the concrete system, with PAC guarantees, up to the horizon $H$.
Let us now discuss how to extend the guarantees to infinite horizon properties.

% Further, 
First, we notice that the dynamics of the concrete system define the probability of sampling a finite trace $\ty_\ell$ at any given time: 
% Some traces may have a probability of being witnessed that decreases or increases with time. 
%
let us recall Example \ref{exa:ell-seq-infinity}. 
The probability of sampling a trajectory exhibiting the sub-sequence $\ty_\ell = y_5 y_1$ as an \emph{initial} subsequence is rather small ($\lambda$) but 
it increases with $H$, until it becomes 1.
%
% The fact that the probability of witnessing $\ty_\ell$ changes over time renders the extension to infinite horizon challenging.
%
Let us define the probability of encountering a trace $\ty_\ell$ at any given time $t$, as 
\begin{multline*}
    \mathbb{P} [ \,
        x_0
        : 
        \mathcal{B}^\omega(x_0)\models\Diamond \ty_\ell
    \, ]
    = 
    \mathbb{P} \left[
    \bigcup_{i=0}^\infty \Pre^i_\mathcal{D}( [ \ty_\ell ] )
    \right] 
    \\
    \geq 
    \mathbb{P} \left[
    \bigcup_{i=0}^k \Pre^i_\mathcal{D}( [ \ty_\ell ] )\right],
\end{multline*}
where 
$$
\Pre^i_\mathcal{D}(S) := \{x'\in\mathcal{D} \ | \ f^i(x') \in S, \> f^j(x')\in\mathcal{D} \text{ for all } 0\leq j\leq i \},
$$
i.e. the points in the domain whose trajectory remains within $\cD$ for all steps $j \leq i$, and the $i$-th step is within $S$. 
% $$
% \Pre(x') = \{ x \in \cD : f(x) = x' \},
% $$
% and $\Pre^i$ is the $i$-th application of $\Pre$. 
For simplicity we denote, for an arbitrary set $S$:
\begin{equation}
\label{eq:mu-definition}
    \mu_0^k(S) := 
    \mathbb{P} \left[
    \bigcup_{i=0}^k \Pre^i_\mathcal{D}( S )
    \right].
\end{equation}
In words, the probability of sampling an initial state $x_0$ that eventually leads to witness the trace $\ty_\ell$ corresponds to the probability of the equivalence class of $\ty_\ell$ together with all the sets \emph{eventually} leading to it, i.e. the $\Pre^k([\ty_{\ell}])$. 
This quantity grows with $k$: whilst for a small $k$ the probability of sampling $\ty_\ell$ may be negligible, for $k \to \infty$, it may reach a significantly large value.
% we quantify the probability of sampling an initial state such that the  trajectory stemming from it includes $\ty_\ell$ at any point in time. 
%
To improve the probability of sampling $\ty_\ell$, we might increase the horizon $k$; whilst this increases the chance of sampling $\ty_\ell$, it also increases the computational needs of our procedure.
We use $\mu(S)$ in place of $\mu_0^0(S)$ to denote the probability measure of a set $S$.
We then introduce the following assumption.
\begin{assumption}
\label{ass:measure-attractive}
% Consider system \eqref{eq:deterministic-sys}. 
Given a system \eqref{eq:sys-probabilistic-init}, assume that a monotonically non-decreasing function $\varphi$ is known for some $k\in\mathbb{N}$, such that for all sets $S$ corresponding to arbitrary unions of equivalence classes, i.e. $S = \bigcup\limits_{j\in J}[\ty_{\ell_j}]$ with $J\subseteq\{1,2,...,|\cY|^{\ell} \}$, the following holds
\begin{equation}\label{eq:finite-to-infinte-measure}
    \mu_0^k(S)
    \geq 
    \varphi(k)
    \cdot
    \mu_0^\infty (S).
\end{equation}
This trivially implies that 
%for all $S = [\ty_{\ell_i}]$, 
    \begin{equation*}
        \mu_0^k(S)< \epsilon \implies \mu_0^\infty(S)<\overline{\gamma} 
        := 
        \dfrac{1}{\varphi(k)} \cdot 
        \epsilon.
    \end{equation*}
\end{assumption}
In practice Assumption \ref{ass:measure-attractive} allows to link the probability measure of visiting any $\ell$-sequence's equivalence class $S$ in $k$ steps with the probability of visiting it in an infinite number of steps. Importantly, the function $\varphi(k)$ describes how the proportion of the measure assigned to the set $\bigcup_{i=0}^k \Pre^i_\mathcal{D}( S )$ with respect to the measure of $\bigcup_{i=0}^\infty \Pre^i_\mathcal{D}( S )$ changes over time, due to the dynamics of the system. Note that the function $\varphi$ necessarily needs to be monotonically non-decreasing. Being able to obtain such a bounding function $\varphi$ requires some a-priori knowledge of properties of the dynamics and the equivalence classes defined by the output map of $\Sigma$. 

Assumption \ref{ass:measure-attractive} allows us to extend the previous results %of Section \ref{subsec:behav-finite-time} 
to behaviours of arbitrary length. 
\begin{proposition}
\label{prop:behav-incl-star}
    Consider $N$ trajectories of length $H$ collected from \eqref{eq:sys-probabilistic-init},
    the data-driven S$\ell$-CA  $\cS^N_\ell$ based on the corresponding $\ell$-sequences in the sampled trajectories, and, for a given $\beta$, the bound $\epsilon(s^*, N, \beta) = \overline{\epsilon}$ resulting from solving program \eqref{eq:scenario-worst-case-problem}.  
    Let Assumption \ref{ass:measure-attractive} hold with $k = H-\ell$ and $\overline{\gamma}$ as defined therein.
    Then,
    for a new initial state $x_0$ sampled from $\mathcal{P}(\cD)$, 
    \begin{equation}
    \label{eq:scenario-simulation}
        \mathbb{P}^N [
    \mathbb{P}[ \,
    \mathcal{S}_{\Sigma_s} \preceq_{\mathcal{B}} \cS^N_\ell
    ]
     \geq 1 - \overline{\gamma}
     \,
    ]
    \geq 1 - \beta, 
    \end{equation}
    i.e. with confidence $1-\beta$,
    $\cS^N_\ell$ probably behaviourally includes the model \eqref{eq:sys-probabilistic-init} with probability not smaller than  $1 - \overline{\gamma}$.
    % \hfill $\square$
\end{proposition}
\begin{proof} 
Given $N$ i.i.d. scenarios, 
let us denote $\theta^*_N$ as the optimal solution of the corresponding scenario program and 
$V(\theta^*_N)$ as the violation probability. 
$V(\theta^*_N)$ represents the probability of drawing a new initial condition $x_0$ which results in a $H$-long behaviour exhibiting one (or more) previously unseen $\ell$-sequence, i.e.
\begin{equation*}
    V(\theta^*_N) := 
    \mathbb{P}
    [x_0  : 
    \cB_H(\cS_{\Sigma_s}(x_0)) \models \Diamond\ty_\ell 
    \wedge \ty_\ell\notin\cX_\ell^N 
    ], 
\end{equation*}
where $\ty_\ell$ represents any sub-sequence of length $\ell$ in $\cB_H(\cS_{\Sigma_s}(x_0))$, the $H$-long behaviour of $\cS_{\Sigma_s}$ starting from $x_0$. 
The scenario theory assures that the violation is upper-bounded by $\bar{\epsilon}$, i.e. $\overline{\epsilon} > V(\theta^*)$,  with confidence not smaller than $1-\beta$.
Denote by $\tilde{S}$ the set of unseen $\ell$-sequences $\tilde{\mathtt{y}}_{\ell_j}$, such that $\cB^\omega(\cS_{\Sigma_s})$ exhibits $\tilde{\mathtt{y}}_{\ell_j}$ but $ \tilde{\mathtt{y}}_{\ell_j}\notin\cX_\ell^N$. 
%
% Consider the set $S = [\tilde{\ty}_\ell]$: 
The scenario theory ensures that 
\begin{equation*}
    V(\theta^*_N) 
    =
    \mathbb{P}
    [
    x_0  :  \cB_H(x_0) \models \Diamond\tilde{\mathtt{y}}_{\ell_j}, \tilde{\mathtt{y}}_{\ell_j}\in\tilde{S}
    ]
    = \mu_0^{H-\ell}
    \left( \bigcup\limits_{\tilde{\mathtt{y}}_{\ell_j}\in\tilde{S}} [\tilde{\mathtt{y}}_{\ell_j}]\right),
\end{equation*}
where $\mu^{H-\ell}_0(\cdot)$ is the probability measure of all initial conditions $x_0 \in \cD$ which exhibit any $\tilde{\ty}_{\ell_j} \in \tilde{S}$ in at most $H$ steps. Equivalently, $x_0$ must visit any $[\tilde{\ty}_{\ell_j}]$ in at most $H-\ell$ steps.
% where the exponent $(H-\ell)$ is derived 
% equal to the number of steps needed to create an $H$-long behaviour starting from any given $\ell$-sequence, employing the domino rule.
% where the exponent $(H-\ell)$ denotes the largest number of time steps during which a behaviour starting from $x_0$ can visit the equivalence class $[\tilde{y}_{\ell_j}]$.
% since a behaviour of length $H$ contains $(H-\ell + 1)$ $\ell$-sequences. As we are taking the $\Pre$ of an $\ell$-sequence, it takes $H-\ell$ steps to obtain a trajectory of length $H$.  
%
Let us apply Assumption \ref{ass:measure-attractive}, with $k = H-\ell$. It holds that 
\begin{equation*}
    \mu_0^\infty\left( \bigcup\limits_{\tilde{\mathtt{y}}_{\ell_j}\in\tilde{S}} [\tilde{\mathtt{y}}_{\ell_j}]\right)
    \leq 
    \frac{1}{\varphi(k)}
    \cdot \mu_0^k
    \left( \bigcup\limits_{\tilde{\mathtt{y}}_{\ell_j}\in\tilde{S}} [\tilde{\mathtt{y}}_{\ell_j}]\right)
    < 
    \frac{1}{\varphi(k)}
    \cdot \overline{\epsilon} 
    = \overline{\gamma}.
\end{equation*}
It follows that with confidence $1-\beta$, all unseen sequences $\tilde{\ty}_{\ell_j}$ have a total probability measure of being exhibited by an infinite behaviour upper-bounded by $\overline{\gamma}$, hence the behavioural relationship between the concrete model and abstraction holds with PAC bound $\overline{\gamma}$. 
\hfill $\square$
\end{proof}

\smallskip

We can think of the scenario bound $\overline{\epsilon}$ as a bound on the probability measure of encountering an unseen  $\ell$-sequence (with confidence $1-\beta$) until time $H$. Assumption \ref{ass:measure-attractive} allows us to extend the scenario guarantees from the finite to the infinite horizon.

%%%%%%%%%%%%%%%%%%%%%%%%%%%%%%%%%%%%%%%%%%%%%%%%%%%%%%%%%

\subsection{Systems Accepting a Bisimulation Relation}
\label{ssec:bisim}

A trivial consequence of Assumption \ref{ass:measure-attractive} is the following corollary.

\begin{corollary}\label{cor:finit=infinite}
Assume there exists $\overline{k}$ such that  $\varphi(\overline{k}) = 1$ for $\overline{k}<\infty$, then for the data-driven S$\ell$-CA $\cS_\ell^N$ constructed from $N$ trajectories of length $H =\overline{k} + \ell$ Proposition \ref{prop:behav-incl-star} holds with $\overline{\gamma} = \overline{\epsilon}$.
\end{corollary}

Corollary \ref{cor:finit=infinite} gives a direct solution to extend finite time scenario guarantees to infinite time scenario guarantees since it provides the simulation horizon necessary to capture the transient behaviours. This follows from the observation that if $\varphi(\overline{k})=1$ then $\mu_0^{\overline{k}}([\ty_{\ell_i}]) = \mu_0^{\infty}([\ty_{\ell_i}])$ for every $\ty_{\ell_i}$.
Corollary \ref{cor:finit=infinite} is useful to provide guarantees on a broad family of models, i.e. dynamical systems accepting a bisimulation relation.

Let us assume the existence of a deterministic S$\ell$-CA $\cS_{\oell}$ for the embedding of the concrete system $\cS_{\Sigma}$: the equivalence relation $\eqref{eq:equivalence-class}$ for $\ell = \overline{\ell}$ defined on $\cS_{\Sigma}$ is a bisimulation relation if and only if $\cS_{\overline{\ell}}$ is deterministic, see \cite[Ch.~1]{belta2017formal}. %\red{[andrea] added something similar before prop4. maybe merge them?}
%Let us assume system \eqref{eq:deterministic-sys} admits a bisimilar model. Then there exists a length $\overline{\ell}$ such that the corresponding S$\overline{\ell}$-CA  is a bisimulation of the concrete system; further, the S$\overline{\ell}$-CA is deterministic, i.e. every state has only one outgoing transition \cite[Ch.~1]{belta2017formal}. 
We exploit this result to construct a data-driven S$\overline{\ell}$-CA that probably behaviourally includes the concrete system for infinite horizon, as the next proposition shows.

\begin{proposition}\label{prop:bisimulation}
Suppose there exists a deterministic S$\ell$-CA $\cS_{\overline\ell}$ for the embedding of system \eqref{eq:deterministic-sys} $\cS_{\Sigma}$ with $\ell=\overline{\ell}$, and consider the data-driven S$\ell$-CA $\cS_{\overline\ell}^N$  constructed from $N$ trajectories. 
Proposition \ref{prop:behav-incl-star} holds with $\overline{\gamma} = \overline{\epsilon}$ if either
\begin{enumerate}
    \item $H=|\cY|^{\overline{\ell}-1} + \oell - 1$,
    \item if \ $\overline{\ell}\leq H < |\cY|^{\overline{\ell}-1} + \oell - 1 $ and $\cS_{\overline\ell}^N$ is non-blocking.
\end{enumerate}
\end{proposition}
\begin{proof}
See Appendix \ref{app:proof-bisim}. 
\hfill $\square$
\end{proof}
This proposition allows to link the number of steps needed to build a bisimulating abstraction, i.e. $\overline{\ell}$, with the length of the trajectories needed to obtain a data-driven S$\ell$-CA that probably behaviourally includes the concrete system. 
% \begin{example}
% \blue{
% Assume a deterministic \slca, as depicted in} Fig.\ref{fig:chain-automa}. The "chain" of states represents an example of a deterministic \slca \ with the maximum number of states, with $|\cY|^{\ell-1}$ states and a diameter, i.e. the longest path connecting any two states, of $\overline{k} = |\cY|^{\ell-1}-1$. 
% This is also the maximum number of steps to compute the backward-reachable set of any state $s$ in the \slca: thus, $\mu^{\overline{k}}([s]) = \mu^{\infty}([s])$.
% \end{example}
% \begin{figure}
%     \centering
%     \input{models/chain-automa}
%     \caption{Example of \emph{chain} automaton, with $\cY = \{a, b, c\}$ and $\oell=3$. The state space has cardinality $|\cY|^{\oell-1} = 9$, and the diameter is 8.}
%     \label{fig:chain-automa}
% \end{figure}
% \mm{I commented out the example, which fitted here better, to save space.}
Note that, knowing $\overline{\ell}$ still remains a non-trivial task for general models. 
%\blue{Finally, let us discuss the assumption of the existence of a deterministic S$\ell$-CA: since $\cS_{\Sigma}$ is non-blocking and $\cS_{\overline{\ell}}$ is deterministic the equivalence relation $\eqref{eq:equivalence-class}$ for $\ell = \overline{\ell}$ defines a bisimulation relation between $\cS_{\overline{\ell}}$ and $\cS_{\Sigma}$, see \cite[Ch.~1]{belta2017formal}.} \red{[andrea] added something similar before prop4. maybe merge them?}
%
% The next section 
We now demonstrate the validity of Assumption $\ref{ass:measure-attractive}$ for the class of linear systems.

% \blue{discussion tuesday meeting: 
% -- assume we do not know barH. how to estimate it?
% use scenario to ensure we reach barH?
% -- mass of unseen sequences: connection with Valiant for infinite behaviours
% }

%%%%%%%%%%%%%%%%%%%%%%%%%%%%%%%%%%%%%%%%%%%%%%%%%%%%%%%%%%%%%

\subsection{Uncertain Affine Stable Systems}
\label{ssec:affine}

Let us consider the class of affine stable systems, that can be written without loss of generality as
\begin{equation}
\label{eq:affine-sys}
    \Sigma(x):=
    \begin{cases}
        x_{k+1} = f(x_k)=Ax_k+b, 
        \\ 
        y_k = h(x_k),
        \\
        x_0 = x,
        \end{cases}
\end{equation}
where both $A$ and the equilibrium $x_{\text{eq}}$ of the map $f$ are unknown. Note that $b$ can be computed from the relation $b=(I-A)x_{\text{eq}}$. We assume $A$ to be full-rank and to know a bound on its eigenvalues, i.e. 
$||A||_2\leq\alpha<1$ (where $||~\cdot~||_2$ denotes the induced $2$-norm).
%Let us define $\rho \geq |\text{det}(A^{-1})| > 1$. 

For this class of systems, the following proposition holds. 
\begin{proposition}
\label{prop:affine-sys}
    For the affine stable system \eqref{eq:affine-sys}, where two scalars $\rho,\alpha\in\mathbb{R}$ such that $0<||A||_2\leq\alpha<1$, $\rho \geq |\text{det}(A^{-1})| > 1$ are known, and $x_0\sim\mathcal{U}_\cD$, if there exists an equivalence class $[y^*]$ such that the equilibrium $x_{\text{eq}}$ belongs to its interior, Assumption \ref{ass:measure-attractive} holds for every $k\in\mathbb{N}$, with 
\begin{equation}
\label{eq:measure-dyn-linear-stable-system}
    \varphi(k) = 
    \begin{cases}
    \min  \left( \psi(k), \rho^{k-\overline{k}}\right) 
    & \text{for }  k \leq \overline{k}-1,
    \\
    1 
    & \text{for } k \geq \overline{k},
    \end{cases} 
\end{equation}
where 
\begin{equation}
    \psi(k) = 
    \left(1 + \rho^{\overline{k}-1-k}\sum_{i=0}^{z(k)-1}\rho^{-i(k+1)}\right)^{-1},
\end{equation}
and 
where $\overline{k}=\left\lceil \log_{\alpha}\left( \frac{d_{min}}{d_{max}} \right) \right\rceil$, $z(k) = \lceil \overline{k}/(k+1) \rceil-1$, $d_{\text{min}}$ and $d_{\text{max}}$ are respectively the radius of the largest ball entirely contained in $[y^*]$ and the radius of the smallest ball entirely containing $\cD$ centered at $x_{\text{eq}}$.
\end{proposition}
\begin{proof}
See Appendix \ref{app:proof-affine-sys}. 
\hfill $\square$
\end{proof}
Let us offer some intuition on the above proposition. 
Denote by $S_o$ the partition containing $x_{eq}$ and
define:
\begin{equation*}
    \overline{k} = 
    \min_{T} \ x(T) \in S_o \text{ for all } t \geq T, \text{ for all } x \in \cD,
\end{equation*}
i.e. the time instant for which all trajectories reach and remain within $S_o$. 
Proposition~\ref{prop:affine-sys} captures the intuitive fact that $\overline{k}$ must depend on the ratio between the volume of $S_o$ and the volume of $\cD$, as well as on the contractivity rate, i.e. the eigenvalues, of \eqref{eq:affine-sys}. 
Notice that $\overline{k}$ exists and it is finite, as $\cD$ is finite and the system is asymptotically stable. 
Hence, the set $\Pre^k(S_o)$ at time  $k=\overline{k}$ covers the whole domain, i.e. $\cD \subseteq \Pre^{\overline{k}}(S_o)$. 

A similar reasoning applies to unstable affine systems: defining again $S_o$ as the partition containing the equilibrium point, with output symbol $y^*$. All trajectories starting at any point in $\cD \setminus S_o$ exit the domain in a finite number of steps, denoted $\overline{k}$. On the other hand, all trajectories starting in $S_o$  generate $\ell$-sequences that start with $y^*$: we denote them as $y^* \sigma_{\ell-1}$, where $\sigma_{\ell-1}\in\cY^{\ell-1}$ represent any possible output symbol sequence of length $\ell-1$. This allows us to easily establish that $\varphi(k)=1,\;\forall k\geq\overline{k}$.
One can envision a similar generalisation to classes of non-linear dynamics, based on contractivity properties, which we leave as future work. 
\begin{remark}
Note that the assumed property enabling these extensions to infinite horizons, in the previous classes of systems, amount to bounding the possible length of transients of the systems (without exact knowledge of the system itself).    
\end{remark}

%%%%%%%%%%%%%%%%%%%%%%%%%%%%%%%%%%%%%%%%%%%%%%%%%%%%%%%%%

\section{Experimental Evaluation}
\label{sec:experiments}

\subsection{Linear Stable System}
We first consider a linear stable model to show our approach. The system evolves according to 
\begin{equation*}
    x_{k+1} = A \, x_k, 
    \ \text{ where }
    A =
    \frac{1}{3}
    \begin{bmatrix}
    1 & 2 
    \\
    -1 & 1
    \end{bmatrix}.
\end{equation*}
The state space $\cD = [-1,1]^2$ is partitioned into 81 regions by a uniform grid. Note that in order to include the equilibrium with a squared partition we need to assess that $||A||_\infty\leq 1$. 
We sample $N=10^4$ initial states from the uniform distribution $\mathcal{U}_\mathcal{D}$ and let the trajectories run until horizon $H=4$. Let us consider $\ell$-sequences with $\ell=2$. 
We collect $189$ $\ell$-sequences and construct the corresponding abstraction. 
We set the confidence value to $\beta=10^{-12}$ and compute the scenario bounds for \eqref{eq:scenario-simulation-time-H}, which result in 
\begin{equation*}
    \epsilon(s^*, \beta, N) = \overline{\epsilon} =7.45 \cdot 10^{-9}. 
\end{equation*}
In order to extend the guarantees from horizon $H=4$ to the infinite horizon we employ Assumption \ref{ass:measure-attractive}. 
We compute $\rho = |\text{det}(A)^{-1}|$ and $\alpha=||A||_2$, which give\footnote{See Appendix \ref{app:proof-affine-sys} for a detailed derivation.} 
$\overline{k} = 9$. Since $H=4$, we evaluate the function $\varphi$ for $k=H-\ell=2$ 
and, using \eqref{eq:measure-dyn-linear-stable-system}, we get $\varphi(k)=4.57\cdot 10^{-4}$. We can thus guarantee that the abstraction holds for an infinite horizon. %as outlined in Section \ref{sec:infinite-behav}. 
% By Assumption \ref{ass:measure-attractive}, 
Hence, the measure of $\mu_0^\infty(S)$ is upper bounded by
\begin{equation*}
    \mu_0^\infty(S) = \mu_0^{\overline{k}}(S) < \bar{\gamma} = \frac{1}{\varphi(k)} \overline{\epsilon} = 1.62\cdot10^{-5}.
\end{equation*}
Our abstraction holds for infinite horizon properties with PAC guarantees
\begin{equation*}
    \mathbb{P}^N[ 
    \mathbb{P}[\cS \preceq_\cB \cS^N_\ell] \geq 1 - 1.62\cdot 10^{-5}
    ] \geq 1 - 10^{-12}.
\end{equation*}

\subsection{One-dimensional Hybrid System}
\label{subsec:contractive-hybrid-sys}

Let us consider the system \eqref{eq:contracting-sys} from Example \ref{exa:ell-seq-infinity} with $\lambda=10^{-2}$ and let us collect $N$ trajectories from it. 
The output of each partition $P_i$ is denoted $y_i$ for $i = [1, 5]$. 

Trivially, for $\ell=1$, the possible 1-sequences are $\{y_i \}^5_{i=1}$. The probability of sampling each $y_i$ is simply
\begin{equation*}
    \mathbb{P}[\cB_1(x_0) \models y_i] = 2^{-i}, \quad i \in [1, 5].
\end{equation*} 
Let us consider $\ell=2$; the possible 2-sequences and the probability of being sampled are reported in \eqref{eq:contracting-sys-horizon-2} where the smallest non zero probability value $\lambda$ is attained by sequence $y_5 y_1$.

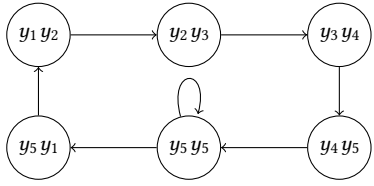
\begin{figure}
    \centering
    \input{models/slca-build}
    \caption{Non deterministic \slca \  for the model \eqref{eq:contracting-sys}, with $\ell=2$.}
    \label{fig:contractive-slca}
\end{figure}

We collect $N=10^4$ trajectories of horizon $H = 2$,
which provide all the $2$-sequences: hence, we know that the resulting S$\ell$-CA is a valid abstraction for the concrete system. 
We construct an abstraction with a state space $\mathcal{X}_\ell^N = \{ y_1 y_2, y_2 y_3, y_3 y_4, y_4 y_5, y_5 y_1, y_5 y_5\}$, shown in Fig.\ref{fig:contractive-slca}. 
%
% We can now compute the PAC bounds for the system behaviours until time $H = 2$.
%
Let us set a confidence value $\beta=10^{-12}$; the scenario program has a complexity $s^*_N = 3$, and by inserting this value in \eqref{eq:scenario-simulation} we get 
\begin{equation*}
    \epsilon(s^*_N, \beta, N) = 4.80 \cdot 10^{-3},
\end{equation*}
thus the data-driven abstraction holds for horizon $H=2$ with PAC bound $\epsilon = 4.8 \cdot 10^{-3}$. To prove that the bounds hold we compute the exact probability of violation, assuming full knowledge of the system: by taking into account the probability of each sequence, it is possible to show that 
\begin{equation*}
    \mathbb{P}[V(\theta^*_N) \leq \epsilon] = \mathbb{P}[ \,
    \mathbb{P}[\mathcal{S} \preceq_{\mathcal{B}} \cS^N_\ell
    ]\geq1-\epsilon]\approx 1-2.25\cdot10^{-44},
\end{equation*}
hence, we show that the scenario bounds we obtain are conservative.

Let us now consider the infinite behaviours of system \eqref{eq:contracting-sys}. \newline
Analysing the $\Pre$ sets of the equivalence classes, we find  $\overline{k} = 7$:
this value represents the time horizon $\overline{k}$ such that $\mu^{\overline{k}}_0(S) = \mu^{\infty}_0(S) = 1$ for all sets $S$ corresponding to a (union of) equivalence classes. As discussed in Corollary \ref{cor:finit=infinite}, by sampling trajectories of length $H=\overline{k}+\ell= 9$ we can provide guarantees for the infinite horizon behaviours.
We repeat the experiment considering $N = 10^4$ trajectories with $\ell=2$ and $H$ as above. Again, we collect the 6  ${\ell}$-sequences shown above, which can be obtained from $s^*_N = 1$ trajectory; we get a new scenario bound 
\begin{equation*}
    \epsilon(s^*_N, \beta, N) = 3.47 \cdot 10^{-3}.
\end{equation*}
%
%

%%%%%%%%%%%%%%%%%%%%%%%%%%%%%

\subsection{Path Planning}

We consider a path planning problem, as depicted in Fig.\ref{fig:path_plan}, where an agent lies within a  $10 \times 10$ grid state space. 
It is tasked to reach the green target area -- with coordinates $[7, 8] \times [7, 9]$ -- whilst avoiding the obstacles (shown in red) and remaining within  the borders of the state space. 
The agent's initial state is chosen uniformly at random within the white area of the state space, and it  can choose among four actions (up, down, left, right) at every time step, in order to reach the target area.
First, we run a standard Q-learning algorithm \cite{watkins1992q} to train the agent, with time horizon $H=40$.

After the training, we use the newly synthesised control policy for a continuous-space experiment, where the agent can take positions over the continuous $[0, 10]^2$ domain, and its actions are obtained as a weighted average of the actions corresponding to the closest grid points. Formally, the action $a(x)$ results 
% \begin{equation*}
$
    a(x) = \sum_{d(x,g)<1} w_d \cdot a(g),
$
% \end{equation*}
where $x$, $g$ are the locations in the continuous and grid space, respectively, $d(x, g)$ is the distance between points $x$ and $g$, and $w_d$ is a coefficient depending on the $d(x,g)$ -- the weights $w_d$ sum up to 1.
We sample the system and collect the agent's position in terms of $W, R, G$ labels (white, red, and green, respectively). 
We obtain $N=10^4$ trajectories with horizon $H=40$, and consider subsequences of length $\ell=27$. 
We collect 27 different $\ell$-sequences, which can be obtained from a total of  $s^*_N=3$ trajectories. 
By setting the confidence to $\beta = 10^{-12}$, the scenario bound evaluates at 
\begin{equation*}
    \epsilon(s^*_N, \beta, N) = 4.06 \cdot 10^{-3}.
\end{equation*}
We construct the data-driven \slca \ with the 27 $\ell$-sequences, and we verify a safety property: the system always reaches (and remains within) the target set avoiding the obstacles.
Hence,  with confidence $1-\beta$,  the concrete model always reaches the target set with probability greater or equal to $1-\epsilon$ (for a horizon $H = 40$). 

\begin{figure}
    \centering
    \vstretch{.8}{\includegraphics[width=\linewidth, trim={1cm 0.7cm 1cm 1cm}, clip=true]{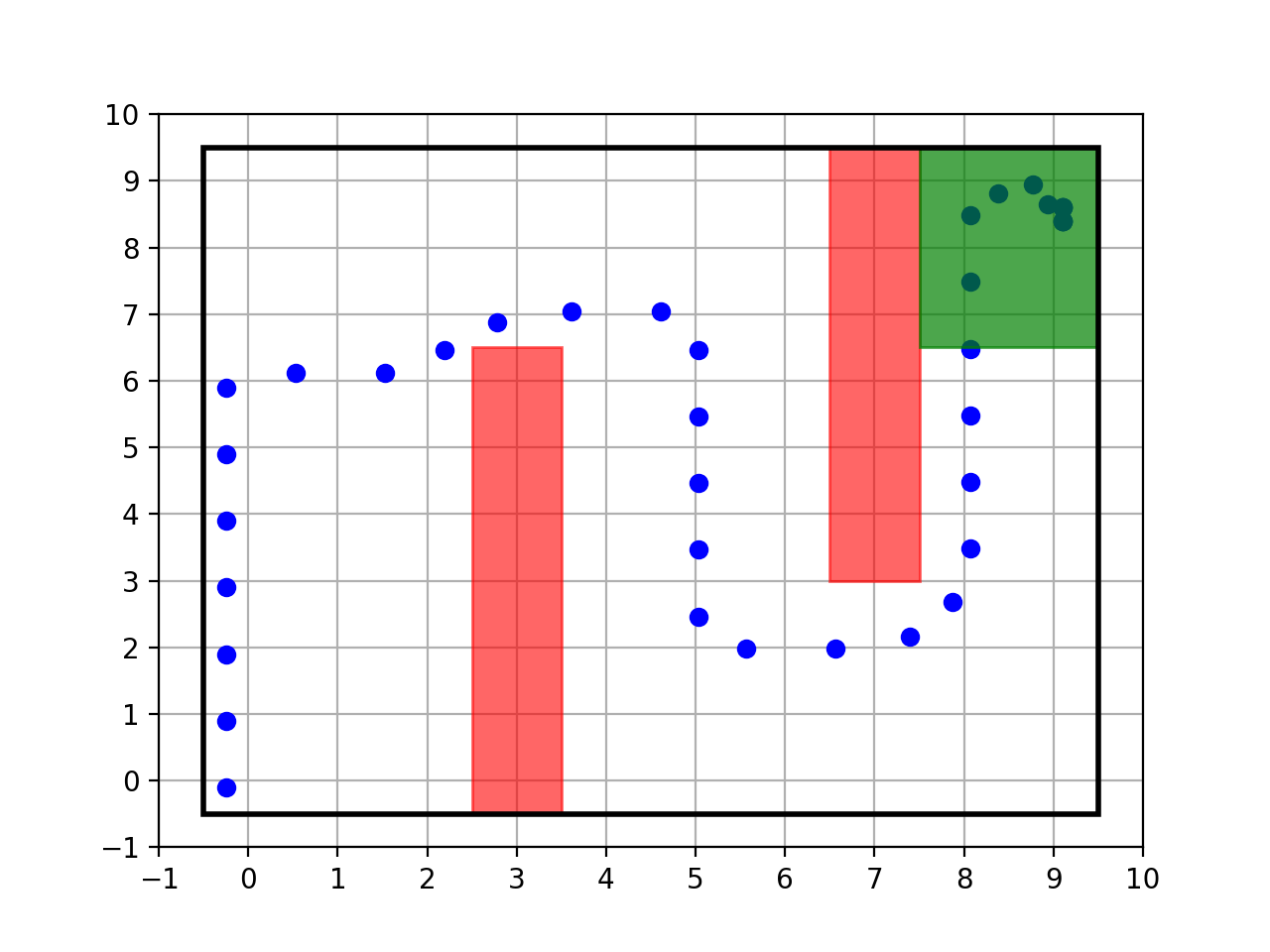}}
    \caption{Example trajectory of the path planning example (blue) with obstacles (red) and target area (green).}
    \label{fig:path_plan}
\end{figure}

%%%%%%%%%%%%%%%%%%%%%%%%%%%%%%%%%%%%%%%%%%%%%%%%%%%%%%%%%%%%%

\section{Conclusions and Future Work}
\label{sec:conclusion}

We have presented a method to 
construct a finite, data-driven abstraction of a deterministic system with unknown dynamics.
We introduce the notion of probabilistic behavioural inclusion, and use it to bound the probability of unseen behaviours of the concrete system. 
We then build 
an $\ell$-complete automaton that generates behaviours of the concrete system, based on trajectories up to time $H$ -- which can be useful in machine learning approaches as RL. 
We further prove that 
this construction can be used to simulate the behaviour of the system for infinite time, if the concrete system admits a bisimulation relation.
 
Our approach can be used to verify any logic specification: 
the PAC guarantees provided by our approach rely upon the collected behaviours, and whether these satisfy the desired property. As such, the concrete model complexity (e.g. nonlinearities) is irrelevant. 

%
% We show that the data-driven approach is useful to probabilistically verify properties of linear deterministic system. 
%
The current method finds applications in the verification of unknown systems, whilst future work includes the extension to control synthesis and stochastic systems.

%%%%%%%%%%%%%%%%%%%%%%%%%%%%%%%%%%%%%%%%%%%%%%%%%%%%%%%%%%%%%
\begin{acks}
This work was supported by the European Research Council through the SENTIENT project (ERC-2017-STG \#755953).

We would like to thank Gabriel de Albuquerque Gleizer and Giannis Delimpaltadakis for their helpful discussions throughout the development of this work.
\end{acks}
%%%%%%%%%%%%%%%%%%%%%%%%%%%%%%%%%%%%%%%%%%%%%%%%%%%%%%%%%%%%%
\bibliographystyle{abbrv}
\bibliography{main}

\appendix
\section{Proofs}

\subsection{Proof of Proposition \ref{prop:bisimulation}}
\label{app:proof-bisim}

\emph{Proof of (1).} 
Consider the \emph{deterministic} S$\ell$-CA for $\cS_{\Sigma}$, denoted  $\cS_{\overline{\ell}}$. Since we assume $\cS_{\Sigma}$ to be non-blocking, necessarily $\cS_{\overline{\ell}}$ is also non-blocking.
We begin by bounding the cardinality of the state set of $\cS_{\overline{\ell}}$. 
Given a state $\ty_{\oell_i}$, consider the set of reachable states
\begin{equation*}
    \mathcal{R}(\ty_{\oell_i}) = \{\ty_{\oell}\in\cX_{\oell} \ | \ (\ty_{\oell_i},\ty_{\oell})\in\cE_{\oell}\}.
\end{equation*}
By the domino rule, the first $\oell-1$ elements of the sequence $\ty_{\oell}$ are identical to the last $\oell-1$ elements of $\ty_{\oell_i}$, and by determinism, $|\mathcal{R}(\ty_{\oell_i})| = 1$. 
Thus, there cannot exist two states $\ty_{\oell_j}, \ty_{\oell_k}$ sharing the same first $\oell-1$ elements. We conclude that the maximum cardinality of the set $\cX_{\oell}$ is equivalent to the number of possible prefixes of length $(\oell-1)$ given the output map $\cY$, which amounts to $|\cY|^{\oell-1}$. 
One example of transition system with the maximum number of states is represented by a "chain" of states, i.e. when each state has one outgoing and one incoming transition. 
% Notice that we obtain the maximum cardinality when the transition system emulates a chain of states, as show in Fig.\ref{fig:chain-automa}.
%
Given $\cS_{\oell}$, 
% Given any state $\ty_{\oell_i}\in\cX_{\oell}$, 
we denote with $\overline{k}$ the \emph{diameter} of the underlying graph, i.e. the longest path connecting any two states. 
% longest finite internal behaviour of $\cS_{\overline{\ell}}$ with no repeating states; 
In general  $\overline{k} \leq |\cY|^{\oell-1}-1$, and this is a strict bound.
The diameter of the abstraction affects $\mu^k_0([\ty_{\oell_i}])$ and allows us to invoke Corollary \ref{cor:finit=infinite}.
Indeed, from \eqref{eq:equivalence-class}
\begin{equation*}
    \text{Pre}_\mathcal{D}([\ty_{\oell_i}]) = \text{Pre}_\mathcal{D}([y_{i_1}...y_{i_{\oell}]}) \\
    = \bigcup_{y\in S} [y \ y_{i_1}...y_{i_{\oell}}]
\end{equation*}
where $S = \{y\in\cY \ | \ \text{Pre}_\cD([\ty_{\oell_i}])\cap [y]\neq\emptyset\}$. The determinism of $\cS_{\oell}$ implies that for $y\in S$
\begin{equation*}
 [yy_{i_1}...y_{i_{\oell}}] =  [yy_{i_1}...y_{i_{{\oell}-1}}] = [\ty_{\oell_j}]
\end{equation*}
To see this, suppose there exists $\tilde{\ty}_{\oell_j} = yy_{i_1}...y_{i_{{\oell}-1}}$ such that $[\tilde{\ty}_{\oell_j}]\neq [yy_{i_1}...y_{i_{\oell}}]$. Trivially, from $\eqref{eq:equivalence-class}$, it follows that 
$[yy_{i_1}...y_{i_{\oell}}]\subseteq [\tilde{\ty}_{\oell_j}]$. Since $\cS_{\Sigma}$ is non-blocking and $[\tilde{\ty}_{\oell_j}]\setminus[yy_{i_1}...y_{i_{\oell}}]\neq \emptyset$ it follows that $|\mathcal{R}({\tilde{\ty}_{\oell_j}})|>1$ which is a contradiction.
Thus, we have shown that 
\begin{equation*}
    \text{Pre}_\mathcal{D}([\ty_{\oell_i}]) = \bigcup_{y\in S} [yy_{i_1}...y_{i_{\oell-1}}]
\end{equation*}
or, in words, that the image of the operator $\text{Pre}_\cD$ on any equivalence class $[\ty_{\oell_i}]$ consists of a union of equivalence classes of other $\oell$-sequences.
Hence, for any $\ty_{\oell_i }$
% It follows that for any $\ty_{\oell_i }$
\begin{equation*}
    \mu_{0}^{k}([\ty_{\oell_i }]) = \mathbb{P} \left[
    \bigcup_{i=0}^k \Pre^i_\mathcal{D}( [\ty_{\oell_i }] )
    \right]= \mu_{0}^{\infty}([\ty_{\oell_i }]), 
    \text{ for } k \geq \overline{k}.
\end{equation*}
By inspection of the $\cS_{\oell}$, we can track the $\Pre$'s of any $[\ty_{\oell_i}]$ simply by following the abstraction's edges in the opposite direction, until we reach either an empty set or a state previously visited; the value $\overline{k}$ denotes the maximum number of steps it takes for this to happen.

Let us now consider the data-driven perspective. It is easy to see that any trajectory of length 
\begin{equation*}
    H = \overline{k} + \oell
    \leq |\cY|^{\oell-1} + \oell -1,
\end{equation*}
ensures to collect all possible $\oell$-sequences that can be generated from a given initial condition $x_0$. Thus, there exists a function $\varphi$, such that $\varphi(\overline{k}) = 1$ holds. 
% where $\overline{H} = |\cY|^{\oell-1} + \oell-2$, which implies $\varphi(\overline{H}) = 1$. 
By Corollary \ref{cor:finit=infinite}, for the data-driven S$\ell$-CA $\cS_{\oell}^N$ constructed from $N$ trajectories of length $H = \overline{k} + \oell$ Proposition $\ref{prop:behav-incl-star}$ holds with $\overline{\gamma} = \overline{\epsilon}$. Note that for this case $\cS_{\oell}^N$ is necessarily non-blocking, therefore domino completion is not needed.

\medskip

% Alternatively, consider the data-driven S$\ell$-CA $\cS_{\oell}^N$ constructed from $N$ trajectories of length $H < \overline{H}$. If $\cS_{\oell}^N$ is blocking, apply domino completion. Note that after domino completion $\cS_{\oell}^N$ might loose the determinism. Recall the violation probability as 
% \begin{equation}
%     V(\theta^*_N) := 
%     \mathbb{P}
%     [x_0  : 
%     \cB_H(\cS(x_0)) \models \Diamond\ty_{\oell} 
%     \wedge \ty_{\oell}\notin\cX_{\oell}^N 
%     ], 
% \end{equation}
% Then \blue{split probability for sequences beginning with already sampled $\oell$-sequences and those which don't}
% \begin{equation}
%     V(\theta^*_N) = \mathbb{P}
%     [x_0  : 
%     \cB_H(\cS(x_0)) \models \Diamond\ty_{\oell} 
%     \wedge \ty_{\oell}\notin\cX_{\oell}^N 
%     \wedge 
%     ]
% \end{equation}

\emph{Proof of (2).} 
By assumption, the concrete system admits a deterministic  abstraction $\cS_{\oell}$, i.e. one outgoing transition per state. 
Again by assumption,  the data-driven $\cS^N_{\oell}$ 
obtained by sampling $N$ trajectories is deterministic and non-blocking. 
Then, the data-driven abstraction can simulate the future behaviours of \emph{all} the sampled $\ell$-sequences -- given the uniqueness of the outgoing transitions. 
Let us now consider the internal behaviour stemming from any unseen $\oell$-sequence $\tilde{\ty}_{\oell}\in\cX_{\oell}$. This can either be disjoint from $\cS^N_{\oell}$ -- i.e. the internal behaviour never reaches any $\ty_{\oell}\in\cX_{\oell}^N$-- or, in finite time, it reaches one of the collected states $\ty_{\oell}\in\cX_{\oell}^N$.
In either situations, the only way to witness $\tilde{\ty}_{\oell}$ is by \emph{sampling} them as \emph{initial} portion of a trajectory. 
This event falls under the guarantees offered by the scenario theory; formally 
\begin{multline*}
     V(\theta^*_N) = 
     \mathbb{P}\left[ x_0  : 
    \cB_H(\cS(x_0)) \models \Diamond \tilde{\ty}_{\oell} 
    \ \wedge \ \tilde{\ty}_{\oell} \notin \cX_{\oell}^N 
    \right] \\
    = \mathbb{P}\left[ x_0  : 
    \cB^\omega(\cS(x_0)) \models \Diamond \tilde{\ty}_{\oell} 
    \ \wedge \ \tilde{\ty}_{\oell} \notin \cX_{\oell}^N 
    \right]\leq \overline{\epsilon},
\end{multline*}
hence the data-driven abstraction probably behaviourally contains the concrete model with bound $\overline{\gamma} = \overline{\epsilon}$, satisfying Proposition \ref{prop:behav-incl-star}. Note that if $\cS^N_{\oell}$ is blocking after collecting $N$ trajectories of length $\overline{\ell}\leq H < |\cY|^{\overline{\ell}-1} + \oell - 1 $ it can be rendered non-blocking by domino completion. In this case $\cS^N_{\oell}$ includes \emph{all} the possible behaviours stemming from the collected $\oell$-sequences, and, in this sense, provides an over-approximation of the system's behaviours.

%

%%%%%%%%%%%%%%%%%%%%%%%%%%%%%%%%%%%%%%%%%%%%%%%%%%%%%%%%%%%%%

\subsection{Proof of Proposition \ref{prop:affine-sys}}
\label{app:proof-affine-sys}

\begin{proof}
Consider $[y^*]$ the equivalence class containing the equilibrium, defined by the output map $h(x)=y^*$. By assumption $x_{\text{eq}}$ is an interior point of $[y^*]$. Define the following quantities
\begin{align*}
    d_{\text{min}} &= \max\{r>0 \ | \ ||x-x_{\text{eq}}||_2 \leq r \implies x\in[y^*]\} \\
    d_{\text{max}} &= \min\{r>0 \ | \ ||x-x_{\text{eq}}||_2 \geq r \implies x\not\in\cD\}
\end{align*}
Since $||A||_2 \leq \alpha < 1$, the set $S_o = \{x\in[y^*] \ | \ ||x-x_{\text{eq}}||_2\leq d_{\text{min}}\  \}$ is such that 
\begin{equation}\label{eq:pre-inclusion}
    S_o\subset f^{-1}(S_o),
\end{equation} or equivalently, since $A$ is invertible, $f(S_o)\subset S_o$.
Moreover, for all $x \in \cD$, 
\begin{equation*}
\| f^{-1} (x-x_{\text{eq}}) \|_2 \geq \alpha^{-1} \| x-x_{\text{eq}} \|_2.
\end{equation*}
Let us consider a set $Q$ such that $Q \cap S_o = \emptyset$. For any point in $Q$, there exists a finite number of steps such that $\Pre^k(Q) \cap \cD = \emptyset$. 
Indeed, for all $x \in Q$ 
\begin{equation*}
\|f^{-k} (x-x_{\text{eq}}) \|_2 > \alpha^{-k} \cdot d_{min} \geq d_{max}, 
\text{ for all } k \geq 
\overline{k} = \left\lceil \log_{\alpha}\left( \frac{d_{min}}{d_{max}} \right) \right\rceil.
\end{equation*}
We conclude that 
\begin{equation*}
f^{-\overline{k}} (Q) \cap \cD = \emptyset, 
\end{equation*}
i.e. $\Pre^{\overline{k}}(Q)$ is entirely outside $\cD$. 
Therefore, recalling \eqref{eq:mu-definition}, it holds that
\begin{equation}
\label{eq:converge-in-H-minus-1}
    \mu^{\overline{k}-1}_0(Q) = \mu^{\overline{k}}_0(Q) = \mu^{\infty}_0(Q). 
\end{equation}
In other words, there exists a finite number of times that the $\text{Pre}(\cdot)$ operator can be applied to any set in $\cD\setminus S_o$ before the image does not intersect with the domain $\cD$. Note that if only a lower bound on $d_{\text{min}}$ and/or an upper bound on $d_{\text{max}}$ are known we obtain an upper bound for $\overline{k}$, hence the same reasoning still applies.

Let us consider the partition $S_o$ containing the equilibrium. 
Combining \eqref{eq:pre-inclusion} with the argument above, it is easy to see that 
\begin{align*}
     f^{-\overline{k}-1}(S_o) & \subset f^{-\overline{k}}(S_o),\\ f^{-\overline{k}}(S_o)\cap\cD & =\Pre^{\overline{k}}(S_o) \cap \cD = \cD
\end{align*}
We have thus proved that $\Pre_{\cD}^{\overline{k}}(S_o) = \cD$, which gives
\begin{equation*}
\label{eq:converge-in-H}
    \mu^{\overline{k}}_0(S_o) = \mu^{\infty}_0(Q). 
\end{equation*}
Additionally, since $f(S_o) \subset S_o$, each point in $S_o$ will generate the infinite behaviour $(y^*)^\omega$. In terms of equivalence classes, this means that 
\begin{multline*}
S_o \subseteq 
\{ x \in \cD \ | \ h(f^i(x)) = y^* \text{ for } 0 \leq i \leq \ell-1 \} = [y^*_{\ell}],
\end{multline*}
for all $\ell > 0$.

Next, we derive the function $\varphi$, under the assumption that $x_0\sim\mathcal{U}_\cD$.
Consider again a set $Q$ such that $Q \cap S_o = \emptyset$.
First, we recall that $\mu^1_0(Q) = 
    \mathbb{P}\left[ Q \cup \Pre_\cD(Q) \right]
    $, and that $\mu_0^0(\Pre_{\cD}(Q)) \leq |\text{det}(A)^{-1}|\mu_0^0(Q) \leq \rho \mu_0^0(Q)$. 
Then, 
\begin{equation*}
    \mu^1_0(Q) = 
    \mathbb{P}\left[ Q \cup \Pre_\cD(Q) \right]
    \leq 
    \mu(Q) (1 + \rho).
\end{equation*}
Let $q,k\in\mathbb{N}^+$ with $q\geq k$ and denote by $z(k) := \lceil (q+1)/(k+1) \rceil-1$. The set $\bigcup_{i=0}^{q}\text{Pre}_{\cD}^i(Q)$ can be recast as a union of sets in the form $\bigcup_{i=t}^{t+k}\text{Pre}_{\cD}^i(Q)$ as
\begin{equation*}
    \bigcup_{i=0}^{q}\text{Pre}_{\cD}^i(Q) = \bigcup_{i = 0}^{w-1}\left(  \bigcup_{j = q - i(k+1)-k}^{q-i\cdot(k+1)} \text{Pre}_{\cD}^{j}(Q)\right) \cup \bigcup_{j = 0}^{k} \text{Pre}_{\cD}^{j}(Q).
\end{equation*}
By the union bound (Boole's inequality) it follows directly that
\begin{equation*}
        \mu_0^q(Q)\leq \sum_{i=0}^{z(k)-1}\mu_{q-i(k+1)-k}^{q-i\cdot(k+1)}(Q) + \mu_{0}^k(Q).
    \end{equation*}
We obtain that
\begin{equation*}
    \mu_0^{q}(Q) \leq \mu_0^{k}(Q)\left(1 + \rho^{q-k}\sum_{i=0}^{z(k)-1}\rho^{-i(k+1)}\right).
\end{equation*}
Using \eqref{eq:converge-in-H-minus-1}, and setting $q=\overline{k}-1$ we use the relation above for time $\overline{k}-1$, to get 
\begin{equation}\label{eq:finite-time-diff-horizon}
    \mu^{\infty}_0(Q) =  \mu^{\overline{k}-1}_0(Q) 
    \leq 
    \mu_0^{k}(Q)\left(1 + \rho^{\overline{k}-1-k}\sum_{i=0}^{z(k)-1}\rho^{-i(k+1)}\right).
\end{equation}
where $z(k) := \lceil \overline{k}/(k+1) \rceil-1$.
Consider now the set $S_o$. Since $S_0\subset f^{-1}S_o$, we have
\begin{equation*}
\label{eq:mu0s0}
    \mu^{k+1}_0(S_o)
    \leq 
    \rho \cdot \mu^k_0(S_o) , 
    \qquad
    \mu^{k+2}_0(S_o)
    \leq 
    \rho^2 \cdot \mu^k_0(S_o) , 
\end{equation*}
and so on. Since $\mu^{\overline{k}}_0(S_o) = \mu^{\infty}_0(S_o)$, then 
\begin{equation}
\label{eq:mu-infty-rho-minus}
    \mu^{\infty}_0(S_o) = 
    \mu^{\overline{k}}_0(S_o) 
    \leq 
    \rho^{\overline{k}-k} \cdot \mu^k_0(S_o) .
\end{equation}
To conclude, we provide a bound for the union of $S_o$ with any other set $Q\in\cD$. We have 
\begin{equation*}
    \mu_0^{\overline{k}}(S_o \cup Q) = \mu_0^{\overline{k}}(S_o) = \mu(\cD) = 1.
\end{equation*}
Therefore 
\begin{equation}\label{eq:mu-infty-unions}
    \mu_0^{k}(S_o \cup Q) \geq \mu_0^{k}(S_o) \geq \mu_0^{\infty}(S_o)\rho^{k-\overline{k}} =\mu_0^{\infty}(S_o \cup Q)\rho^{k-\overline{k}}.
\end{equation}
Thus, by combining \eqref{eq:finite-time-diff-horizon}-\eqref{eq:mu-infty-rho-minus}-\eqref{eq:mu-infty-unions}, 
\begin{equation*}
\label{eq:mu-infty-rho-minimum}
    \mu^k_0(W) \geq \mu^\infty_0(W) \cdot 
    \min \left( \left(1 + \rho^{\overline{k}-1-k}\sum_{i=0}^{z(k)-1}\rho^{-i(k+1)}\right)^{-1}, \rho^{k-\overline{k}} \right).
\end{equation*}
for any $W\in\cD$ as claimed.
\hfill 
$\square$
\end{proof}

%%%%%%%%%%%%%%%%%%%%%%%%%%%%%%%%%%%%%%%%%%%%%%%%%%%%%

\end{document}

%% file: models/example-slca.tex
\tikzset{
        ->,  % makes the edges directed
        >=stealth', % makes the arrow heads bold
        node distance=2cm, % specifies the minimum distance between two nodes. Change if n
        every state/.style={thick, minimum width=1.2cm,fill=gray!0}, % sets the properties for each ’state’ n
        initial text=$ $, % sets the text that appears on the start arrow
        }

\begin{tikzpicture}[scale=0.8, transform shape]
    % \node[state] (q0) {$(1 )$};
    % \node[state, below of=q0] (q1) {$(2 )$};
    % \draw   (q1) edge[loop right] node{} (q1)
    %         (q0) edge[loop right] node{} (q0)
    %         (q1) edge[above, bend right] node{} (q0)
    %         (q0) edge[above, bend right] node{} (q1);
    
    \small
    \node[state] (q0) {$(y_1 y_1 )$};
    \node[state, right of=q0] (q1) {$(y_1 y_2 )$};
    \node[state, below of=q1, xshift=-1cm] (q2) {$(y_2 y_2 )$};
    \draw   (q0) edge[loop above] node{} (q0)
            (q2) edge[loop left] node{} (q2)
            (q1) edge[above] node{} (q2)
            (q0) edge[above] node{} (q1);
    
    \footnotesize
    \node[state, right of=q1, node distance=2.5cm] (q0) {$(y_1 y_1 y_2 )$};
    \node[state, right of=q0, node distance=2.5cm] (q1) {$(y_1 y_2 y_2 )$};
    \node[state, below of=q1] (q2) {$(y_2 y_2 y_2 )$};
    \node[state, below of=q0] (q3) {$(y_1 y_1 y_1 )$};
    \draw   (q2) edge[loop right] node{} (q2)
            (q3) edge[loop left] node{} (q3)
            (q1) edge[above] node{} (q2)
            (q3) edge[above] node{} (q0)
            (q0) edge[above] node{} (q1);
\end{tikzpicture}

%% file: models/grid-part.tex
\tikzset{
        % ->,  % makes the edges directed
        >=stealth', % makes the arrow heads bold
        node distance=2.5cm, % specifies the minimum distance between two nodes. Change if n
        every state/.style={thick, minimum width=1.2cm,fill=gray!0}, % sets the properties for each ’state’ n
        initial text=$ $, % sets the text that appears on the start arrow
        }

\begin{tikzpicture}[scale=0.90, transform shape]
    
    % first system depiction
    
    \draw[step=1.0,gray,thin] (0,0) grid (3,3);
    
    \node[] (p0) at (0.5, 0.5) {\footnotesize $P_0$};
    \node[] (p1) at (1.5, 0.5) {\footnotesize $P_1$};
    \node[] (p2) at (2.5, 0.5) {\footnotesize $P_2$};
    
    \node[] (p3) at (0.5, 1.5) {\footnotesize $P_3$};
    \node[] (p4) at (1.75, 1.25) {\footnotesize $P_4$};
    \node[] (p5) at (2.5, 1.5) {\footnotesize $P_5$};
    
    \node[] (p6) at (0.5, 2.5) {\footnotesize \red{$\mathbf{P_6}$}};
    \node[] (p7) at (1.5, 2.5) {\footnotesize \dgreen{$\mathbf{P_7}$}};
    \node[] (p8) at (2.5, 2.5) {\footnotesize $P_8$};
    
    \node[] (c1) at (0.75, 0.25) {\Large \textbf{$\cdot$}};
    \node[] (c3) at (0.65, 0.15) {\Large \textbf{$\cdot$}};
    \node[] (c2) at (0.85, 0.15) {{\textbf{\Large $\cdot$}}};

    \node[] (dest1) at (1.25, 1.3) {} ;
    
    \draw[->]  (c1) edge[left, bend right] node{} (dest1) ;
    
    \node[blue] (d1) at (1.2, 1.3) {\Large \textbf{$\cdot$}};
    \node[blue] (d3) at (1.25, 1.5) {\Large \textbf{$\cdot$}};
    \node[blue] (d2) at (1.3, 1.35)  {\Large \textbf{$\cdot$}};
    
    \draw[blue] (1.25, 1.5) ellipse (0.2cm and 0.3cm);
    
    % unsafe set
    \node[] (unsafe) at (0.5, 2.5) {};
    
    % from step1 to unsafe
    \draw[->, blue]  (d2) edge[left, bend left] node{} (unsafe) ;
    
    % outside step1 to new state
    \node[] (extra1) at (1.8, 1.5) {\Large \textbf{$\cdot$}};
    \node[] (safe) at (1.5, 2.5) {};
    \draw[->]  (extra1) edge[left, bend right] node{} (safe) ;

    %%% abstraction
    
    \node[draw, circle] (a11) at (5 , 0.5) {};
    \node[draw, circle] (a21) at (5 , 1.5) {};
    \node[draw, circle, thick, red] (a31) at (5 , 2.5) {};
    
    \node[draw, circle] (a12) at (6 , 0.5) {};
    \node[draw, circle] (a22) at (6 , 1.5) {};
    \node[draw, circle, thick, black!40!green] (a32) at (6 , 2.5) {};
    
    \node[draw, circle] (a13) at (7 , 0.5) {};
    \node[draw, circle] (a23) at (7 , 1.5) {};
    \node[draw, circle] (a33) at (7 , 2.5) {};
    
    \draw[->] (a11) edge[left, bend right] node{} (a22); 
    % \draw[->] (a11) edge[left, bend left] node{} (a21); 
    \draw[->] (a22) edge[left, bend right] node{} (a32); 
    \draw[->, blue] (a22) edge[left, bend left] node{} (a31);

\end{tikzpicture}

%% file: models/blocking.tex
\tikzset{
        ->,  % makes the edges directed
        >=stealth', % makes the arrow heads bold
        node distance=2.5cm, % specifies the minimum distance between two nodes. Change if n
        every state/.style={thick, minimum width=1.2cm,fill=gray!0}, % sets the properties for each ’state’ n
        initial text=$ $, % sets the text that appears on the start arrow
        }

\begin{tikzpicture}[scale=0.88, transform shape]
    
    \small
    \node[state, right of=q0, node distance=4.0cm] (q0) {$(y_1y_1 y_2 )$};
    \node[state, right of=q0, node distance=4cm] (q1) {$(y_1y_2y_1 )$};
    \node[draw, dashed, rounded corners, minimum height = 1cm, 
    below of=q1, xshift=1cm, node distance=2.2cm] (q2) {$(y_2y_1y_2 )$};
    \node[draw, dashed, rounded corners, minimum height = 1cm,
    below of=q1, xshift=-1cm, node distance=2.2cm] (q21) {$(y_2y_1y_1 )$};
    
    \node[state, below of=q0, node distance=2.2cm] (q3) {$(y_1y_1y_1 )$};
    
    \draw   (q3) edge[loop left] node{} (q3)
            (q3) edge[left] node{} (q0)
            (q0) edge[above] node{} (q1);
    \draw (q1) edge[left, dashed, bend right] node{} (q2)
          (q1) edge[left, dashed] node{} (q21);
    \draw (q2) edge[bend right, dashed, right] node{} (q1);
    \draw (q21) edge[dashed, above] node{} (q3);
    
\end{tikzpicture}

%% file: models/histogram.tex
\tikzset{
        % ->,  % makes the edges directed
        >=stealth', % makes the arrow heads bold
        node distance=2.5cm, % specifies the minimum distance between two nodes. Change if n
        every state/.style={thick, minimum width=1.2cm,fill=gray!0}, % sets the properties for each ’state’ n
        initial text=$ $, % sets the text that appears on the start arrow
        }

\pgfplotsset{
  compat=newest,
  xlabel near ticks,
  ylabel near ticks
}

\begin{tikzpicture}[scale=0.5, transform shape, font=\Large]
    
    % first system depiction
    
    \draw[gray,thin] (0,0)--(0,5)--(5,5)--(5,0)--cycle;
    \draw[gray,thin] (2,2)--(3,2)--(3,3)--(2,3)--cycle;
    \draw[gray, thin] (0,3)--(2,3)--(2,0);
    \draw[gray, thin] (5,2)--(3,2)--(3,5);
    
    \node[] (aa) at (2.5, 2.5) {$aa$};
    \node[] (aa) at (1, 1.5) {$bc$};
    \node[] (aa) at (1.5, 4) {$cd$};
    \node[] (aa) at (3.5, 1) {$eb$};
    \node[] (aa) at (4, 3.5) {$de$};

\hskip 5pt
%%% histogram
%

    \begin{axis}[
      ybar,
      at={(500, -2)},
      bar width=20pt,
      xlabel={$\ell$-sequence},
      ylabel={Probability},
      ymin=0,
      ytick=\empty,
      xtick=data,
      axis x line=bottom,
      axis y line=left,
      enlarge x limits=0.2,
      symbolic x coords={$aa$,$bc$,$cd$,$de$,$eb$},
      xticklabel style={anchor=base,yshift=-\baselineskip},
      nodes near coords={\pgfmathprintnumber\pgfplotspointmeta\%},
      yscale=0.9
    ]
      \addplot[fill=white] coordinates {
        ($aa$, 4)
        ($bc$, 24)
        ($cd$, 24)
        ($de$, 24)
        ($eb$, 24)
      };
    \end{axis}

\end{tikzpicture}

%% file: models/slca-build.tex
\begin{tikzpicture}[scale=1, transform shape]
    
    \small
    \node[state] (q0) {$y_1 y_2$};
    \node[state, right of=q0, node distance=2.0cm] (q1) {$y_2 y_3$};
    \node[state, right of=q1, node distance=2.0cm] (q2) {$y_3 y_4$};
    \node[state, below of=q2, node distance=1.5cm] (q3) {$y_4 y_5$};
    \node[state, left of=q3, node distance=2.0cm] (q4) {$y_5 y_5$};
    \node[state, left of=q4, node distance=2.0cm] (q5) {$y_5 y_1$};

    \draw[->] (q1) edge[right] node{} (q2);
    \draw[->] (q0) edge[above] node{} (q1);
    \draw[->] (q2) edge[above] node{} (q3);
    \draw[->] (q3) edge[above] node{} (q4);
    \draw[->] (q4) edge[above] node{} (q5);
    \draw[->] (q5) edge[above] node{} (q0);
    
    \draw   (q4) edge[loop above] node{} (q4);
    
\end{tikzpicture}